\documentclass[11pt,a4paper]{article}
\usepackage[fleqn]{amsmath}
\usepackage{amsfonts,amsthm,amssymb}
\usepackage{comment}
\usepackage{mathtools}
\usepackage{subcaption}
\usepackage{url}
\usepackage{multirow}
\usepackage{wrapfig}
\usepackage[pdftex,dvipsnames]{xcolor}  
\allowdisplaybreaks

\renewcommand{\arraystretch}{1.5} 

\newtheorem{theorem}{Theorem}
\newtheorem{proposition}[theorem]{Proposition}
\newtheorem{lemma}[theorem]{Lemma}
\newtheorem{corollary}[theorem]{Corollary}
\newtheorem{definition}[theorem]{Definition}

\newtheorem{observation}[theorem]{Observation}

\usepackage{authblk}

\begin{document}
\title{Computing shortest paths amid non-overlapping weighted disks}
%
%

\date{}

\author[1]{Prosenjit Bose\thanks{Email: jit@scs.carleton.ca}}
\author[2]{Jean-Lou De Carufel\thanks{Email: jdecaruf@uottawa.ca}}
\author[3]{Guillermo Esteban\thanks{Email: g.esteban@uah.es}}
\author[1]{Anil Maheshwari\thanks{Email: anil@scs.carleton.ca}}

\affil[1]{Carleton University, Ottawa, Canada}
\affil[2]{University of Ottawa, Ottawa, Canada}
\affil[3]{Universidad de Alcal\'{a}, Alcal\'{a} de Henares, Spain}

\maketitle

\begin{abstract}
In this article, we present an approximation algorithm for solving the Weighted Region Problem amidst a set of $ n $ non-overlapping weighted disks in the plane. For a given parameter $ \varepsilon \in (0,1]$, the length of the approximate path is at most $ (1 +\varepsilon) $ times larger than the length of the actual shortest path. The algorithm is based on the discretization of the space by placing points on the boundary of the disks. Using such a discretization we can use Dijkstra's algorithm for computing a shortest path in the geometric graph obtained in (pseudo-)polynomial time.
\end{abstract}


\section{Introduction}
\label{sec:intro}

Computing a geodesic path (i.e., shortest path) between two points $ s $ and~$ t $ in a geometric setting is one of the most studied problems in computational geometry. Applications of geometric shortest path problems are ubiquitous, appearing in diverse areas such as robotics~\cite{gaw,rowe,Sharir}, video games design~\cite{kamphuis,sturtevant2}, or geographic information systems~\cite{floriani}.
We refer to Mitchell~\cite{Mitchell1} for an excellent survey on geometric shortest path problems.

In contrast to the classical shortest path problem in graphs, where the space of possible paths is discrete, in geometric settings the space is continuous: the source and target points can be anywhere within a certain geometric domain (e.g., a polygon, the plane, a surface), and the set of possible paths to consider has infinite size.
Many variations of geometric shortest path problems have been studied before, depending on the geometric domain, the objective function (e.g., Euclidean metric, link-distance, geodesic distance), or specific domain constraints (e.g., obstacles in the plane, or holes in polygons).
Finding shortest paths among polygonal obstacles in the plane has drawn great interest~\cite{aleksandrov2000approximation,aleksandrov2005determining,bose2024steiner,hershberger1999optimal,Mitchell1,mitchell1991weighted,sun2001bushwhack}. Some of these results apply directly to real-world problems. For instance, for modelling subdivisions of surfaces, embedding models use cylindrical faces, quadrics or patch together surfaces that are defined via bicubic or quadratic splines, see, e.g., \cite{braid1975synthesis,braid1980stepwise,connolly1987application,lienhardt1991topological,requicha1980representations}. Motion planning problems that need to be solved for the advancement of robotics typically involve motion of curved objects through obstacles having curved boundaries~\cite{chew1985planning,KirkpatrickL16-cccg}.
Modern font design systems rely upon conic and cubic spline curves~\cite{pavlidis1983curve,pratt1985techniques}. Numerous applications need efficient algorithms for processing curved objects directly~\cite{forrest1986invited,smith1986invited}. The way to tackle arbitrary real objects has been to approximate them first as polygons or polyhedra of a sufficient number of vertices for the particular application. This process is generally unsatisfactory, see \cite{dobkin1990computational}. For these reasons, in this paper, as in the works by Chang et al.~\cite{chang2005shortest}, Chen et al.~\cite{chen2013computing}, Chen and Wang~\cite{chen2015computing}, and Hershberger et al.~\cite{hershberger2022near}, we focus on the problem of computing shortest paths among curved objects. In particular, we consider disks of different radii with a non-negative weight assigned to them.

\subsection{Previous results}

One of the most general versions of the shortest path problem that has been studied consists of a subdivision of the two-dimensional space. Without loss of generality, we assume it to be triangulated. Each region has a (non-negative) weight associated to it, representing the cost per unit distance of traveling in that region. Thus, the cost of traversing a region is typically given by the Euclidean distance traversed in the region, multiplied by the corresponding weight. The resulting metric is often called the \emph{weighted region metric}, and the problem of computing a shortest path between two points under this metric is known as the \emph{Weighted Region Problem} (WRP). This problem is very general, since it allows to model many well-known variants of geometric shortest path problems. Indeed, having that all weights are equal makes the metric equivalent to the Euclidean metric (up to scale), while using two different weight values, such as $ 1 $ and $ \infty $, allows to model paths amidst obstacles.

The WRP was first introduced by Mitchell and  Papadimitriou~\cite{Mitchell1,mitchell1991weighted}. They provided an approximation algorithm that computes a $ (1+\varepsilon) $-approximation path in $ O\left(n^8\log{\frac{nNW}{w\varepsilon}}\right) $ time, where $ N $ is the maximum integer coordinate of any vertex of the subdivision, $ W $ (respectively, $ w $) is the maximum (respectively, minimum non-zero) integer weight assigned to a face of the subdivision.

Recently, it has been shown that the WRP cannot be solved exactly within the Algebraic Computation Model over the Rational Numbers (ACM$ \mathbb{Q} $) \cite{de2014note}. In this model one can compute exactly any number that can be obtained from rational numbers by a finite number of basic operations. This emphasizes the need for high-quality approximation paths instead of optimal paths. The result in~\cite{de2014note} provides a theoretical explanation to the lack of exact algorithms for the WRP, and the fact that several authors proposed algorithms for computing approximation paths, reducing the running time, or produced geometric problem instances with fewer ``bad'' configurations (e.g., the Delaunay triangulation is used to maximize the minimum angle).

Most of the results for the WRP are focused on polygonal obstacles. The most common scheme followed in the literature is to discretize the geometric space by positioning Steiner points, and then build a graph by connecting pairs of Steiner points, see~\cite{aleksandrov1998varepsilon,aleksandrov2000approximation,aleksandrov2005determining,bose2024steiner,sun2001bushwhack}. An approximate solution is constructed by finding a shortest path in this graph by using well-known combinatorial algorithms (e.g., Dijkstra's algorithm). See Table~\ref{table:results3} for the time complexity of the approximation algorithms designed following this scheme.

\begin{table}[tb]
\centering
\renewcommand{\arraystretch}{1.4}
\begin{tabular}{|c|c|}
\hline Time complexity  & Reference \\ \hline
$ O\left(n^8\log{\frac{nNW}{w\varepsilon}}\right) $ & \cite{mitchell1991weighted} \\ \hline
$ O\left(N^4\log{\left(\frac{NW}{w\varepsilon}\right)}\frac{n}{\varepsilon^2}\log{\frac{nN}{\varepsilon}}\right) $ & \cite{aleksandrov1998varepsilon} \\ \hline
$ O\left(N^2\log{\left(\frac{NW}{w}\right)}\frac{n}{\varepsilon}\log{\frac{1}{\varepsilon}\left(\frac{1}{\sqrt{\varepsilon}}+\log{n}\right)}\right) $ & \cite{aleksandrov2000approximation} \\ \hline
 $ O\left(N^2\log{\left(\frac{NW}{w}\right)}\frac{n}{\varepsilon}\log{\frac{n}{\varepsilon}}\log{\frac{1}{\varepsilon}}\right) $ & \cite{sun2001bushwhack}        \\ \hline
 $ O\left(N^2\log{\left(\frac{NW}{w}\right)}\frac{n}{\sqrt{\varepsilon}}\log{\frac{n}{\varepsilon}}\log{\frac{1}{\varepsilon}}\right) $ & \cite{aleksandrov2005determining} \\\hline
\end{tabular}
\caption{Some $ (1+\varepsilon)$-approximation algorithms for the WRP. In this table, $ N $ is the maximum integer coordinate of any vertex of the subdivision, $ W $ (resp., $ w $) is the maximum (resp., minimum non-zero) integer weight assigned to a face of the subdivision.}
\label{table:results3}
\end{table}

However, we are aware of only a few publications that treat curved objects. In general, we do not seem to have a good grasp on the complexity of weighted shortest paths when the region boundaries are nonlinear curves. The particular case where we consider $ n $ circular obstacles in the two-dimensional space (i.e., the weight of all the regions is infinity) can be solved exactly by using Dijkstra's algorithm in $ O(n^2 \log{n}) $ time~\cite{chew1985planning}. The standard technique is to compute a combinatorial graph $ G = (V,E) $ whose edges are the common tangents between the obstacles, as well as arcs between pairs of nodes that are consecutive on the boundary of the same region. The vertex set $ V $ comprises the endpoints of the tangent lines in each of the obstacles. Moreover, Kim et al.~\cite{kim2004shortest} devised two filters, an ellipse filter and a convex hull filter, which reduce the space complexity significantly and efficiently. The technique in~\cite{chew1985planning} was improved to $ O(n\log{n}+k) $ time, where $ k $ is the size of the extended visibility graph of the union of the pseudodisks~\cite{chen2013computing}. Later, Chen and Wang~\cite{chen2015computing} computed a shortest path avoiding a set $S$ of $h$ pairwise disjoint splinegons with a total of~$n$ vertices in $O(n + h\log{h} + d)$ time, where $ d $ is a parameter sensitive to the geometric structures of the input, by applying a bounded degree decomposition of the set of obstacles. This improves the result in~\cite{chen2013computing} when $h = o(n)$. In addition, an important algorithmic challenge arises when considering a model where the radius of a set of disks grows over time at some (maximum) speed~\cite{maheshwari2007n,van2008planning}.

Obstacles with curved boundaries present both algebraic and combinatorial challenges~\cite{chang2005shortest}. Thus, Hershberger et al.~\cite{hershberger2022near} proposed an $ O(n\log{n}) $ time algorithm for the shortest path problem based on certain assumption on the computation of locating the intersection of two bisectors defined by pairs of curved obstacle boundary segments. They also provided a $ (1 + \varepsilon) $-approximation of a shortest path in $O\left(n\log{n} + n\log{\frac{1}{\varepsilon}}\right) $ time without the bisector computation assumption.

Moreover, if we consider the case where inside each region we can travel between any pair of points at no cost whereas outside all regions the travel cost between two points is their Euclidean distance, this can be seen as a redefinition of the additively weighted point set spanner problem. Bose et al.~\cite{bose2011spanners} were able to show that it is possible to design a graph $ G $ with a linear number of edges such that for any pair of disks $ D$ and $ D' $ there is a path in $G$ whose
length is arbitrarily close to the Euclidean distance between $D$ and $D'$. Recently, this was improved by Smid~\cite{smid2021improved} by reducing the number of
edges needed by a factor of $4$.

\subsection{Our results}

Sometimes, the shape of a real-world curved object can be approximated using a polygon whose vertices are specified by a subset of $ c $ points on the object, where~$ c $ is a sufficiently large value. Then, one approach to solving the WRP on a set of curved regions would be to approximate each region with a polygon, and then use existing algorithms that work on polygons. However, this method is not always optimal~\cite{dobkin1990computational}.

Let $ \mathfrak{D} = \{D_1, \ldots, D_n\} $ be a set disks, each with a radius $ R_i > 0 $, and for any pair~$ D_j $ and $ D_k $, $ 1 \leq j < k \leq n $, $ D_j \cap D_k = \emptyset $. In addition, each disk $D_i$ has a (non-negative) weight $ \omega_i $ assigned to it. In this paper, we provide an algorithm to compute a path between two points amidst $ \mathfrak{D} $ that is at most $ (1+\varepsilon) $ times larger than the actual shortest path. To solve this problem, we use the traditional technique of partitioning the 2-dimensional space into a discrete space by using a non-trivial Steiner points placement and designing an appropriate graph. Without loss of generality, we may assume that $ s$ and $ t $ are vertices of this graph. In particular, the main results of this paper are:

\begin{itemize}
    \item The special case of the WRP where all the regions are disks having a weight $ \omega = 0 $ or $ \omega \geq \frac{\pi}{2} $ can be solved exactly by using visibility graph techniques and Dijkstra's algorithm in $ O(n^2) $ time. See Section~\ref{sec:weight0infty}.
    \item For the general version of the WRP, we propose a discretization that consists of a set of Steiner points along the boundary of each disk. We first place some vertices, called \emph{vertex vicinity centers}, evenly on the boundary of each disk. Then, if the weight of the disk is strictly positive, we create an annulus around each vertex vicinity center, and we place a set of Steiner points inside each annulus. For a given approximation parameter $ \varepsilon \in (0, 1] $, the number of vertices of the discretization is at most $ C(\mathfrak D)\frac{n}{\varepsilon} $, where $ C(\mathfrak D) $ captures geometric parameters and the weights of $ \mathfrak D $. See Section~\ref{sec:discretization}.
    \item We show that the weighted length of the approximated path between any pair of nodes is at most $ (1 + \varepsilon) $ times the weighted length of a shortest path. This approximation path can be computed by executing shortest path algorithms on the graph formed by Steiner points where two Steiner points are joined by an edge. See Section~\ref{sec:onedisk}.
    \item We also show how to create a linear-sized $ t $-spanner to reduce the running time of the algorithms that compute a weighted shortest path when the disks have any non-negative weight assigned to them. See Section~\ref{sec:spanner}.
\end{itemize}

\section{Preliminaries}
\label{sec:preliminaries}

Any continuous (rectifiable) curve lying in the two-dimensional space is called a \emph{path}. Let $ \Pi(s,t) $ denote a path from a source point $ s $ to a target point $ t $ among the set of disks~$ \mathfrak{D} = \{D_1, \ldots, D_n\} $. Let $ R_i $ and $ c_i $ be, respectively, the radius and the center of each disk $ D_i $. Let $ \omega_i \in \mathbb R_{\ge 0}, \ i \in \{1, \ldots, n\} $, be the weight associated to a disk $ D_i \in \mathfrak{D} $, which represents the cost of traveling a unit Euclidean distance inside that disk. In addition, and without loss of generality, we can assume that the weight outside the disks is $ 1 $. Otherwise, we could always rescale the weights to be $ 1 $ outside the disks. Then, the weighted length of $ \Pi(s,t) $ is given by $ \lVert \Pi(s,t) \rVert = \mu + \sum_{i=1}^n \omega_i \cdot \lvert \pi_i \rvert $, where $ \mu $ denotes the Euclidean length of the intersection between $ \Pi(s,t) $ and the space outside the disks, and $ \lvert \pi_i \rvert $ denotes the Euclidean length of the intersection between $ \Pi(s,t) $ and a disk $ D_i $, that is, $ \pi_i = \Pi(s,t) \ \cap \ D_i $. In case $ \pi_i $ coincides with an arc of $ D_i $, the weight of traveling along that arc is given by $ \min\{1, \omega_i\} $. Given two distinct points $ s $ and~$ t $ in the plane, a weighted shortest path $ \mathit{SP_w}(s,t) $ is a path that minimizes the weighted length between $ s $ and $ t $.

Observe that every path consists of a sequence of (straight or circular-arc) segments whose endpoints $ a_1, \ldots, a_{m} $ are on the boundary of the disks in $ \mathfrak{D} $. These endpoints $ a_1, \ldots, a_m $ are called \emph{bending points}.

We now present some properties of a shortest path between two points on the boundary of the same disk that will be useful in the forthcoming sections. Observation~\ref{obs:length} gives the (weighted) length of a subpath between two points~$ p $ and $ q$ on the boundary of a disk $ D \in \mathfrak{D} $. The result can be proved using the law of cosines.

\begin{observation}
    \label{obs:length}
    Let $ p $ and $ q $ be two consecutive bending points of the path~$ \mathit{SP_w}(s,t) $ on the boundary of a disk~$ D $ centered at~$ c $ with radius~$ R$ and weight~$ \omega \geq 0 $. Let~$ \theta $ be the angle~$ \angle{cpq} $. Then,
    \begin{itemize}
        \item If a shortest path from $ p $ to $ q $ coincides with an arc of $ D $, $ \lVert \mathit{SP_w}(p,q)\rVert = R\cdot(\pi - 2\theta) $; see the red path in Figure~\ref{fig:tboundary}.
        \item If a shortest path from $ p $ to $ q $ only intersects the boundary of $ D $ at $ p $ and $ q $, $ \lVert\mathit{SP_w}(p,q) \rVert = \omega\cdot2R\cos{\theta} $; see the blue path in Figure~\ref{fig:tboundary}.
    \end{itemize}
\end{observation}

\begin{figure}[tb]
    \centering
    \includegraphics[width = 0.5\textwidth]{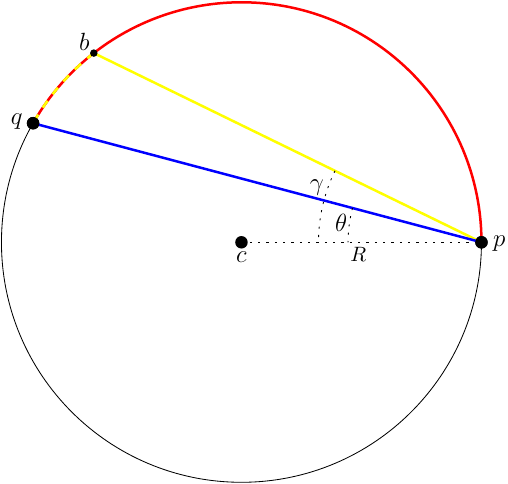}
    \caption{The two possible types of shortest paths between $ p $ and $ q$ on the boundary of a disk are depicted in red and blue.}
    \label{fig:tboundary}
\end{figure}

The following result follows from the fact that the weight of the boundary of a region is given by the minimum among the weights of the two adjacent regions.

\begin{observation}
    \label{lem:weightless1}
    Let $ p $ and $ q $ be two consecutive bending points of the path $ \mathit{SP_w}(s,t) $ on the boundary of a disk~$ D $. Let $ \omega \in [0,1] $ be the weight of~$D $. Then $ \mathit{SP_w}(p,q) $ only intersects the boundary of~$ D $ at $ p $ and $ q $.
\end{observation}

Now consider a special case of the WRP where all the regions have a weight $ \omega \geq \frac{\pi}{2} $. Lemma~\ref{lem:greaterpi2} states that when the weight of a disk is at least $ \frac{\pi}{2}$, then the disk can be considered as an obstacle. Hence, if all disks have weight at least $ \frac{\pi}{2} $, we can use one of the algorithms that compute an exact shortest path between any two pair of points among a set of obstacles.

\begin{lemma}
    \label{lem:greaterpi2}
    Let $ p $ and $ q $ be two consecutive bending points of the path $ \mathit{SP_w}(s,t) $ on the boundary of a disk~$ D $. Let $ \omega \geq \frac{\pi}{2} $ be the weight of~$D $. Then $ \mathit{SP_w}(p,q) $ coincides with a shortest arc of $ D $ from $ p $ to $ q $.
\end{lemma}

\begin{proof}
    We need to prove that the weight of a path intersecting the interior of the disk is at least as large as when going along the boundary, i.e., that $ R\cdot(\pi - 2\theta) \leq 2 R\omega\cos{\theta} \Longleftrightarrow \pi - 2\theta \leq 2\omega\cos{\theta} $, for any angle $ \theta \in \left[0, \frac{\pi}{2}\right) $.

    We know that $ \omega \geq \frac{\pi}{2} $, so $ 2 \omega\cos{\theta} \geq \pi\cos{\theta} $. Thus, it is sufficient to prove that $ \pi - 2\theta \leq \pi\cos{\theta} $. We first minimize the function $ \pi\cos{\theta}+2\theta $:

    \begin{equation*}
        \frac{\partial \left( \pi\cos{\theta}+2\theta\right)}{\partial \theta} =-\pi\sin{\theta}+2 = 0 \Longleftrightarrow \sin{\theta} = \frac{2}{\pi}.
    \end{equation*}

    For this value of $ \theta $, we get that $ \pi\cos{\theta}+2\theta \geq \sqrt{\pi^2-4} + 2\arcsin{\frac{2}{\pi}} $, which is greater than~$ \pi $, which gives us the desired result.
\end{proof}

Now, we state that there are no other ways for a shortest path between~$ p $ and $ q $ to intersect the disk~$ D $ than the ones described in Observation~\ref{obs:length}. This means that if $ p $ and $ q $ are two consecutive bending points on the boundary of a disk $ D$, a shortest path between $ p $ and $ q $ is either the straight-line segment between $ p $ and $ q $, or a shortest arc of $ D $ from $ p $ to $ q $. Hence, a shortest path between $ p $ and $ q $ does not bend on the boundary of $ D $. 

\begin{lemma}
    \label{lem:pathnotconsidered}
    Let $ p $ and $ q $ be two consecutive bending points of the path $ \mathit{SP_w}(s,t) $ on the boundary of a disk~$ D $. If $ \omega \in (1,\frac{\pi}{2}) $, then a shortest path between $ p $ and $ q $ is a (straight or circular-arc) segment.
\end{lemma}

\begin{proof}
    Suppose there is a point $ b \neq p, q $ on the boundary of $ D $ where $ \mathit{SP_w}(p,q) $ bends, and let $ \gamma $ be the angle $ \angle{cpb}$, see Figure~\ref{fig:tboundary}. In this case, the length of $ \mathit{SP_w}(p,q) $ is $ \lVert \mathit{SP_w}(p,q) \rVert = R\frac{\sin{(\pi-2\gamma)}}{\sin{\gamma}}\omega+2R(\gamma-\theta) = 2R\cdot\left(\frac{\cos{\gamma}\sin{\gamma}}{\sin{\gamma}}\omega+(\gamma-\theta)\right) = 2R\left(\omega\cos{\gamma}+(\gamma-\theta)\right) $.
    
    The value $ \lVert \mathit{SP_w}(p,q) \rVert $ is minimized when $ \cos{\gamma} = \frac{\sqrt{\omega^2-1}}{\omega} $:
    
    \begin{equation*}
        \frac{\partial \ \lVert \mathit{SP_w}(p,q) \rVert}{\partial \gamma} = -2R\omega\sin{\gamma}+2R = 0 \Longleftrightarrow \omega\sin{\gamma} = 1 \Longleftrightarrow \begin{cases}
            \sin{\gamma} = \frac{1}{\omega}, \\
            \cos{\gamma} = \frac{\sqrt{\omega^2-1}}{\omega}.
        \end{cases}
    \end{equation*}
    
    We can see that the equation holds since $ \omega \in \left(1, \frac{\pi}{2}\right) $. Hence, for this value of~$ \gamma $, the weighted length of $ \mathit{SP_w}(p,q) $ is:

    \begin{equation*}
        \lVert \mathit{SP_w}(p,q) \rVert = 2R\left(\omega\cos{\gamma}+(\gamma-\theta)\right) = 2R\sqrt{\omega^2-1}+2R\left(\arcsin{\left(\frac{1}{\omega}\right)}-\theta\right).
    \end{equation*}
    
    Now, we need to compare the length of $ \mathit{SP_w}(p,q) $ with the length of (i) a path $ \pi_1(p, q) $ that only intersects the boundary of $ D$ at $ p $ and $ q$, and (ii) a path $ \pi_2(p, q) $ along the boundary of $ D$.

    We first define the function $ \omega\cos{\theta}+\theta $ that allows us to prove that $ \lVert \pi_1(p,q) \rVert \leq \lVert \mathit{SP_w}(p,q) \rVert $. The maximum value of the function, when $ \theta \in \left(0, \frac{\pi}{2}\right) $ is obtained next:
    
    \begin{equation*}
        \frac{\partial \left(\omega\cos{\theta}+\theta\right)}{\partial \theta} = -\omega\sin{\theta}+1 = 0 \Longleftrightarrow \omega\sin{\theta} = 1 \Longleftrightarrow \sin{\theta} = \frac{1}{\omega}, \ \theta = \arcsin{\left(\frac{1}{\omega}\right)}.
    \end{equation*}

    Hence,
    
    \begin{align*}
        \omega\cos{\theta} + \theta & \leq \omega\cos{\left(\arcsin{\left(\frac{1}{\omega}\right)}\right)}+\arcsin{\left(\frac{1}{\omega}\right)} =  \sqrt{\omega^2-1}+\arcsin{\left(\frac{1}{\omega}\right)} \\
        \Longrightarrow \omega\cos{\theta} & \leq \sqrt{\omega^2-1}+\arcsin{\left(\frac{1}{\omega}\right)}-\theta \\
        \Longrightarrow 2R\omega\cos{\theta} & \leq 2R\sqrt{\omega^2-1}+2R\left(\arcsin{\left(\frac{1}{\omega}\right)}-\theta\right) \\
        \Longrightarrow \lVert \pi_1(p,q) \rVert & \leq \lVert \mathit{SP_w}(p,q) \rVert.
    \end{align*}

    For $ \theta = 0 $ we need to prove that
    
    \begin{align*}
        \lVert \pi_1(p,q) \rVert = 2R\omega & \leq 2R(\omega\cos{\gamma}+\gamma) = \lVert \mathit{SP_w}(p,q) \rVert \\
        \omega & \leq \omega\cos{\gamma}+\gamma\\
        \omega & \leq \frac{\gamma}{1-\cos{\gamma}}.
    \end{align*}

    Now, we would like to obtain the minimum value of the function $ \frac{\gamma}{1-\cos{\gamma}} $ when $ 0 \leq \gamma \leq \frac{\pi}{2} $:

    \begin{align}
        \label{eq:gamma}
        \frac{\partial \ \frac{\gamma}{1-\cos{\gamma}}}{\partial \gamma} = \frac{1-\cos{\gamma}-\gamma\sin{\gamma}}{(1-\cos{\gamma})^2} & = 0 \nonumber\\
        1-\cos{\gamma} & = \gamma\sin{\gamma}.
    \end{align}

    The only solution to Equation (\ref{eq:gamma}), when $ 0 \leq \gamma \leq \frac{\pi}{2} $, is $ \gamma = 0 $. However, the minimum of the function $ \frac{\gamma}{1-\cos{\gamma}} $ is obtained for $ \gamma = \frac{\pi}{2} $. The value of the function for this value of $ \gamma $ is $ \frac{\pi}{2} $, and the maximum value $ \omega $ can take is $ \frac{\pi}{2} $. Hence, $ \omega \leq \frac{\gamma}{1-\cos{\gamma}} $, and the result is proven.

    Now, we define another function $ \omega^5 -14\omega^3+37\omega$ that allows us to prove that $ \lVert \pi_2(p,q) \rVert \leq \lVert \mathit{SP_w}(p,q) \rVert $. The maximum of this function, when $ \omega \in (1, \frac{\pi}{2}) $, is obtained next:
    
    \begin{equation*}
        \frac{\partial \left(\omega^5 -14\omega^3+37\omega\right)}{\partial \omega} = 5\omega^4-42\omega^2+37 = 0 \Longleftrightarrow \omega = 1.
    \end{equation*}
    
    Hence,
    
    \begin{align*}
        \omega^5 -14\omega^3+37\omega \leq 24 & \Longrightarrow \frac{\omega^4-14\omega^2+37}{24} \leq \frac{1}{\omega} \\
        \Longrightarrow \frac{24-12\omega^2+12+\omega^4-2\omega^2+1}{24} \leq \frac{1}{\omega} & \Longrightarrow 1-\frac{\omega^2-1}{2}+\frac{(\omega^2-1)^2}{24} \leq \frac{1}{\omega}.
    \end{align*}

    The Taylor series of the function $ \cos{x} $ is $ \sum_{n=0}^{\infty} \frac{(-1)^n}{(2n)!}x^{2n}$, for all $ x $. Thus, $ \cos{\sqrt{\omega^2-1}} \leq 1-\frac{\omega^2-1}{2}+\frac{(\omega^2-1)^2}{24} $, and we get that
    
    \begin{align*}
        & \cos{\sqrt{\omega^2-1}} \leq \frac{1}{\omega} \Longrightarrow \sin{\left(\frac{\pi}{2}-\sqrt{\omega^2-1}\right)} \leq \frac{1}{\omega} \\
        \Longrightarrow & \frac{\pi}{2}-\sqrt{\omega^2-1} \leq \arcsin{\left(\frac{1}{\omega}\right)}         \Longrightarrow \pi -2\sqrt{\omega^2-1} \leq 2\arcsin{\left(\frac{1}{\omega}\right)} \\
        \Longrightarrow & R\cdot(\pi -2\theta) \leq 2R\sqrt{\omega^2-1} + 2R\arcsin{\left(\frac{1}{\omega}\right)} -2R\theta \\
        \Longrightarrow & \lVert \pi_2(p,q) \rVert \leq \lVert \mathit{SP_w}(p,q) \rVert.
    \end{align*}

    We proved that the length of both paths $ \pi_1(p,q) $ and $ \pi_2(p,q) $ is not larger than the length of the path $ \mathit{SP_w}(p,q) $. Hence, a shortest path from $ p $ to $ q $ is either the straight-line segment from $ p $ to $ q $ or a shortest arc of $ D $ from $ p $ to~$ q $.
\end{proof}

\section{$ 0/1/\infty$-weighted regions}
\label{sec:weight0infty}

In this section we consider the special case where the weight of each disk $ D_i \in \mathfrak{D} $ is $ \omega_i \in \{0,\infty\} $. As in~\cite{gewali1988path}, we refer to regions with weight $ \infty $ as \emph{obstacles}, and we can think of regions with weight $ 0 $ as places where we can travel at infinite speed. We present an efficient exact algorithm for computing a weighted shortest path avoiding the obstacles.

We begin by recalling that the computation of the shortest path among curved obstacles without the presence of $0$-regions is already known, see, e.g.,~\cite{chen2013computing,chen2015computing}. However, our approach here is to generalize the results by Chen et al.~\cite{chen2013computing} to $0$- and $ \infty$-regions by building a special kind of ``visibility graph'' using a constant number of vertices on the boundary of each disk. Our method is similar to the one by Chen et al., so we refer the reader to their paper for most of the proofs.

\subsection{Summary of our approach}

Using the algorithm of Pocchiola and Vegter in~\cite{pocchiola1995computing}, the visibility graph $ G_{\mathcal{V}} $ of a set of $ n $ pairwise disjoint convex objects in the plane can be constructed in $ O(k + n\log{n}) $ time, where $ k $ is the number of arcs of the visibility graph. However, $ G_{\mathcal{V}} $ may have~$ \Omega(k) $ vertices, and $k$ may be $\Theta(n^2)$. Consequently, running Dijkstra’s algorithm in~$ G_{\mathcal{V}} $ would take $O(k\log{n})$ time, which is a bottleneck.

Our strategy is to transform $ G_{\mathcal{V}} $ to a coalesced graph $ G_{\mathcal{V}}^c $ such that (i) $ G_{\mathcal{V}}^c $ has $O(n)$ vertices and $O(k)$ edges, and (ii) a weighted shortest path in $ G_{\mathcal{V}}^c $ corresponds to a shortest path in $ G_{\mathcal{V}} $. Then we can compute a weighted shortest path in~$ G_{\mathcal{V}}^c $ in $ O(k + n\log{n})$ time. To build $ G_{\mathcal{V}}^c $, we determine a set of distinguished points. Note that the approach in~\cite{chen2013computing} also builds a coalesced graph based on a set of distinguished points. However, we do not need all of the distinguished points defined in~\cite{chen2013computing} since our regions are disjoint, and we also consider disks with weight $ 0 $. 
Our approach produces only $O(n)$ distinguished points. Consequently, the number of vertices in our coalesced graph is also bounded by $O(n)$. Our approach for computing distinguished points is based on the Voronoi diagram $ \mathit{VD}(\mathfrak D)$ of the disks of $\mathfrak D$. Using $\mathit{VD}(\mathfrak D)$, for each disk~$ D $ of $ \mathfrak D$, we create distinguished points on the boundary of $D$ whose number is proportional to the number of neighboring cells of $D$ in $\mathit{VD}(\mathfrak D)$. Since the Voronoi diagram has $O(n)$ cell adjacencies, the total number of such distinguished points is $O(n)$. Since $ G_{\mathcal{V}}^c $ has $O(n)$ vertices and $O(k)$ edges, a shortest path in $ G_{\mathcal{V}}^c $ can be found in $O(n\log{n} + k)$ time.

\subsection{The visibility graph}

First note that we can treat $0$-regions as if they were obstacles~\cite{gewali1988path}. The reason for this is that the path between two points on the boundary of a disk with weight $ 0 $ using the straight-line segment, and an arc of the disk between them have the same weighted length. Thus, we have no need for segments through $0$-regions, and we can treat the $0$-regions as if they were obstacles.

Let $ u $ and $ v $ be two points on the boundary of two disks $ D $ and $ D' $, respectively, where $ D $ and $ D' $ are not necessarily different. We define the path $ \delta(u,v) $ as the weighted shortest path from $ u $ to $ v $ when $ D $ and $ D' $ are the only two disks in the space.

The vertex set of the extended visibility graph $ G_{\mathcal{V}} $ consists of all endpoints of the common tangents between all pairs of disks. Let $ u $ and $ v $ be a pair of vertices of $ G_{\mathcal{V}} $, then we create an edge between two tangency points $ u $ and $ v $ in $ G_{\mathcal{V}} $ if $ \delta(u,v) $ does not intersect the interior of any disk. We also add a circular-arc edge between consecutive vertices on the boundary of the disks. We finally assign to each edge $ (u,v) \in G_{\mathcal{V}} $ the weight $ \lVert \delta(u,v) \rVert $.

Now, we transform $ G_{\mathcal{V}} $ to a directed weighted graph $ G_{\mathcal{V}}^d $ using the results of~\cite{chen2013computing}. We first replace each undirected edge with two directed edges with opposite directions (see the new joints in $ G_{\mathcal{V}}^d $ between an arc edge and a tangent edge in~\cite[Figure $3.2$]{chen2013computing}). Then we assign to each directed edge $ (u,v) $ the weight~$ \lVert \delta(u,v)\rVert $.

\subsection{Coalesced graph}

To define our coalesced graph $ G_{\mathcal{V}}^c $, we introduce a set of $ O(n) $ distinguished points following the notation of Chen et al.~\cite{chen2013computing}. Let $ \mathit{VD}(\mathfrak D) $ be the Voronoi diagram of all disks in $ \mathfrak{D} $. In $ \mathit{VD}(\mathfrak D) $, every disk $ D $ induces a cell $ C(D) $ that contains $ D $. We say that a disk $ D $ is a \emph{neighbor} of another disk $ D' $ if $ C(D) $ and~$ C(D')$ are neighboring in $ \mathit{VD}(\mathfrak D) $. See Figure~\ref{fig:distinguished} for some illustrations of these definitions. For every disk $ D $, with respect to each neighboring disk $ D'$ of $ D $ in $ \mathit{VD}(\mathfrak D) $, we add the following four types of distinguished points to $ D $:
\begin{itemize}
    \item Type (1): there are at most four common tangents between $ D $ and $ D' $; then each tangent point $ x $ on $ D $ of these common tangents is a distinguished point. See Figure~\ref{fig:type1}.
    \item Type (2): the closest point $ y $ on $ D $ to $ D' $ (using the Euclidean distance) is a distinguished point. See Figure~\ref{fig:type2}.
    \item Type (3): there are at most two tangent lines of $ D $ passing through each Voronoi vertex $ v $ of $ C(D) $; then each tangent point $ z $ on $ D $ of these tangent lines is a distinguished point. See Figure~\ref{fig:type3}.
    \item Type (4): the closest point $ w $ on $ D $ to each Voronoi vertex $ v $ of $ C(D) $ (using the Euclidean distance) is a distinguished point. See Figure~\ref{fig:type4}.
\end{itemize}

\begin{figure}[tb]
    \centering
    \begin{subfigure}[b]{0.4\textwidth}
        \centering
        \includegraphics{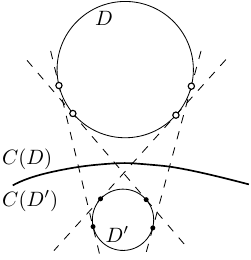}
	    \caption{Type (1) distinguished points.}
	    \label{fig:type1}
    \end{subfigure}
    \quad
    \begin{subfigure}[b]{0.4\textwidth}
        \centering
        \includegraphics{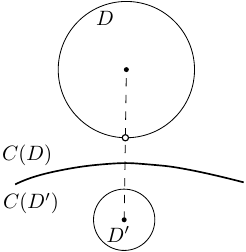}
	    \caption{Type (2) distinguished point.}
	    \label{fig:type2}       
     \end{subfigure}
    \quad
    \begin{subfigure}[b]{0.4\textwidth}
   	\centering
        \includegraphics{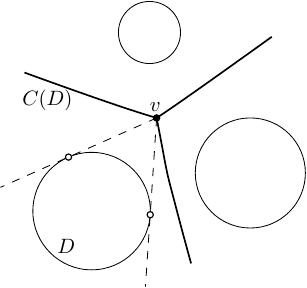}
	    \caption{Type (3) distinguished points.}
	    \label{fig:type3}
    \end{subfigure}
    \quad
    \begin{subfigure}[b]{0.4\textwidth}
   	\centering
        \includegraphics{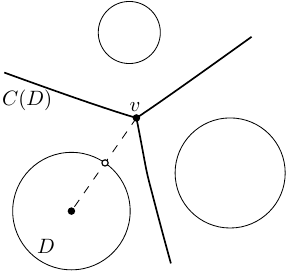}
	    \caption{Type (4) distinguished point.}
	    \label{fig:type4}
    \end{subfigure}
    \caption{Distinguished points on $ D $ are depicted as white disks.}
    \label{fig:distinguished}
\end{figure}

Finally, we add to each disk more distinguished points (if needed) so that the tangent directions at any two consecutive distinguished points along the disk differ by at most $ \frac{\pi}{2}$. The reason for adding the latter points is made explicit in the proof of \cite[Lemma~$ 3.4 $]{chen2013computing}.

Based on the directed graph $ G_{\mathcal{V}}^d $, and the definition of the coalesced graph~$ G_{\mathcal{V}}^c $  in~\cite{chen2013computing}, we define our coalesced graph $ G_{\mathcal{V}}^c $ as follows. Our distinguished points divide the boundary of the disks in $ \mathfrak D $ into arc intervals. Each interval has two directions in which a path may traverse, so they are called \emph{directed intervals}. The set of all such directed intervals forms the vertex set of $ G_{\mathcal{V}}^c $. For each (directed) tangent edge $ e(u, v) $ of $ G_{\mathcal{V}}^d $, with two tangent points $ u $ and $ v $ on disks $ D $ and $ D' $, respectively, there must be a unique directed interval $ (u', u'') $ (resp., $ (v', v'')$) containing $ u $ (resp., $v$) on $D$ (resp., $D'$) such that the path $ u' \rightarrow u \rightarrow v \rightarrow v'' $ is forward going, see Figure~\ref{fig:directedinterval}. The directed interval $(u', u'')$ (resp., $(v', v'')$) is referred to as the \emph{host interval} of $u$ (resp., $v$). 

\begin{figure}[tb]
    \centering
    \includegraphics[width = 0.5\textwidth]{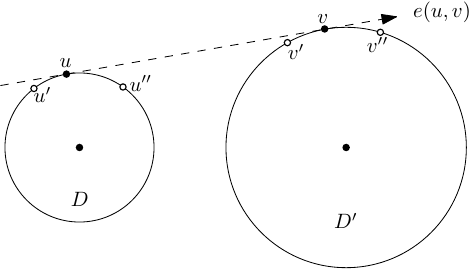}
    \caption{The path $ u' \rightarrow u \rightarrow v \rightarrow v'' $ is forward going. The distinguished points $ u', u'', v' $ and $ v'' $ are depicted as white disks.}
    \label{fig:directedinterval}
\end{figure}

We then create a tangent edge in $ G_{\mathcal{V}}^c $ to link the vertex defined by the interval $(u', u'')$ to the vertex defined by $(v', v'')$, and set its weight to $ \lVert \delta(u,v)\rVert -\lVert \delta(u,u'')\rVert + \lVert \delta(v,v'')\rVert $. Next, we consider the arc edges of $ G_{\mathcal{V}}^c $. For a disk $ D $, suppose $(u', u')$, $(u', u'')$, and $(u'', u'')$ are three consecutive directed intervals on the boundary of $ D $ with the same direction, where $(u', u'')$ is an open interval and the other two are closed intervals. Then, we put an arc edge in $ G_{\mathcal{V}}^c $ from the vertex defined by $(u', u')$ to the vertex defined by $(u', u'')$ with weight $ \lVert \delta(u',u'')\rVert $, and put an arc edge in $ G_{\mathcal{V}}^c $ from the vertex defined by $(u', u'')$ to the vertex defined by $(u'', u'')$ with weight $0$. This completes the definition of our coalesced graph $ G_{\mathcal{V}}^c $.

The next property of the distinguished points shows that the weight of the edges of the coalesced graph $ G_{\mathcal{V}}^c $ is non-negative.

\begin{proposition}
    \label{prop:lengthchain}
    For any disk $ D $, if there is a tangent edge $ e(u,v) $ with the endpoint $ u $ on $ D $ and $ u $ is not a distinguished point, assuming that $ u $ is between two consecutive distinguished points $ u' $ and $ u'' $ on $ D $ such that the subarc $ \widehat{u''u} $ and the tangent edge $ e(u,v) $ do not form a convex chain, then $ \lVert \delta(u,v)\rVert \geq \lVert \delta(u,u'')\rVert $.
\end{proposition}

Before giving the proof of Proposition~\ref{prop:lengthchain}, we first make a helpful observation about the distance between two tangency points. Following the notation of Chen et al., let $ \phi_{p}(D_1) $ be the closest point on a disk $ D_1 $ to a point $ p $. Let $ u \in D $ and $ v \in D' $ be two tangency points. Let $ v_1 $ be the point where $ \delta(u,v) $ enters disk $ D' $, and let $ u_1 $ be the point where $ \delta(u,v) $ leaves disk $ D $. In the following discussion, we let $ t = \phi_{v_1}(D) $. See Figure~\ref{fig:lemma3} for two figures with the notation used for the proof of Proposition~\ref{prop:lengthchain}.

\begin{lemma}
    \label{lem:bounddistance}
    If a disk $ D $ has a distinguished point, say $ a $, on the arc $ \widehat{ut} $, then $ \lVert\delta(u,u'')\rVert \leq \lVert\delta(u,v)\rVert $.
\end{lemma}

\begin{proof}
    If $ D $ is a $ 0 $-region, then $ \lVert\delta(u,u'')\rVert = 0 \leq \lVert\delta(u,v)\rVert $. Hence, we now focus on the case where $ D $ is a $ \infty $-region.
    
    We first prove $ \lVert\delta(u,v)\rVert \geq \lVert\delta(u,t)\rVert \Longleftrightarrow \lVert\delta(v,v_1)\rVert + \lvert \overline{v_1u_1}\rvert + \lVert\delta(u_1,u)\rVert \geq \lVert\delta(u,u_1)\rVert + \lVert\delta(u_1,t)\rVert $. If $ D' $ is a $ \infty $-region $ v = v_1 $ and $ \lVert\delta(v,v_1)\rVert = 0 $, otherwise $ \lVert\delta(v,v_1)\rVert = 0 \cdot \lvert\widehat{vv_1}\rvert = 0 $. Hence, we need to prove that  $ \lvert \overline{v_1u_1} \rvert + \lVert\delta(u_1,u)\rVert \geq \lVert\delta(u,u_1)\rVert + \lVert\delta(u_1,t)\rVert \Longleftrightarrow \lvert \overline{v_1u_1} \rvert \geq \lVert\delta(u_1,t)\rVert $. Let $ p $ be the intersection point of the segment $ \overline{v_1u_1}$ and the tangent line to $ D $ from $ t $, see Figure~\ref{fig:lemma3}. 
    
    \begin{figure}[tb]
    \centering
    \begin{subfigure}[b]{0.4\textwidth}
        \centering
        \includegraphics{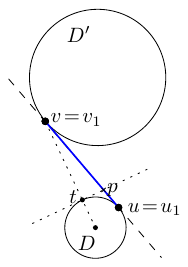}
	    \caption{Disk $ D' $ has weight $ \infty $.}
	    \label{fig:case1}
    \end{subfigure}
    \quad
    \begin{subfigure}[b]{0.4\textwidth}
        \centering
        \includegraphics{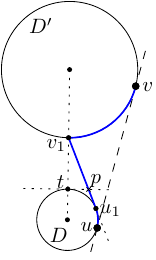}
	    \caption{Disk $ D' $ has weight $ 0 $.}
	    \label{fig:case2}       
     \end{subfigure}
    \caption{The path $ \delta(u,v) $ is represented in blue.}
    \label{fig:lemma3}
\end{figure}
    
    Then $ |\overline{tp}| + |\overline{pu_1}| \geq \lVert\delta(u_1,t)\rVert $. Since the segment $ \overline{v_1t}$ is perpendicular to the tangent line to $ D $ from $ t $, we have $ |\overline{pt}| \leq |\overline{pv_1}| $. Therefore, $ |\overline{u_1v_1}| = |\overline{v_1p}| + |\overline{pu_1}| \geq |\overline{pt}| + |\overline{pu_1}| \geq \lVert\delta(u_1,t)\rVert $.

    Since $ a $ is on $ \widehat{ut} $, $ \lVert\delta(u,v)\rVert \geq \lVert\delta(u,t)\rVert \geq \lVert\delta(u,a)\rVert$. Because $ a $ is a distinguished point on $ D $ and $ u'' $ is the closest distinguished point on $ D $ to $ u $ along the reversing direction, it must be $ \lVert\delta(u,a)\rVert \geq \lVert\delta(u,u'')\rVert $. Thus, $ \lVert\delta(u,u'')\rVert \leq \lVert\delta(u,v)\rVert $.
\end{proof}

\subsubsection{Outline of the proof of Proposition~\ref{prop:lengthchain}}

For simplicity of discussion we only give details on the aspects of the proof that differ from those in the proof of Chen et al. The proof of Proposition~\ref{prop:lengthchain} consists of several cases. In each case, the goal is to find a distinguished point of $ D $ on the arc $ \widehat{ut} $, and then by Lemma~\ref{lem:bounddistance}, $ \lVert\delta(u,u'')\rVert \leq \lVert\delta(u,v)\rVert $.

First of all, let $ D_v $ be the disk containing a tangency point $ v $. Since $ D_v $ and~$ D $ are in different cells of $ \mathit{VD}(\mathfrak D) $, $ \delta(u,v) $ intersects a Voronoi edge $ h $ of the cell $ C(D) $, and let $ i $ be this intersection point. Let $ D' $ be the neighbor of $ D $ sharing the edge $ h$ with $ D$. Let us denote the bisector curve between $ D $ and $ D' $ as $ B $, and let $ b'' $ be the intersection point between $ B $ and the segment $ \ell $ between the centers of $ D $ and $ D' $. The line containing $ \ell $ divides the plane into two half-planes, let $ B_r $ be the portion of $ B $ contained in the halfplane that we enter first when moving from the intersection point of $ \ell $ and $ D $ in clockwise direction around $ D $, and let $ B_\ell $ be the other portion. See Figure~\ref{fig:notationbisector} for an illustration of all these definitions.

\begin{figure}[tb]
    \centering
    \includegraphics[width = 0.4\textwidth]{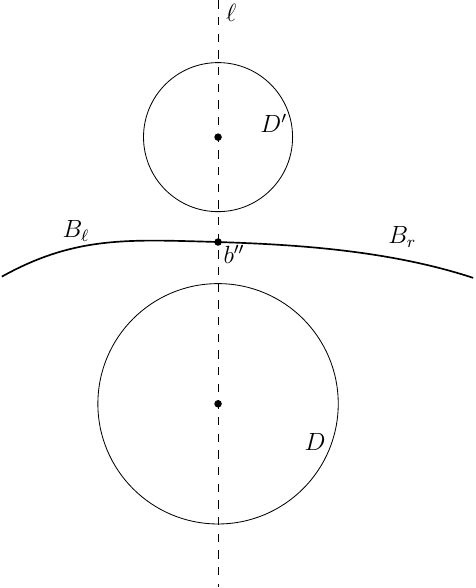}
    \caption{Illustrating the definitions of $ \ell $, $ b'' $, $ B_r $ and $ B_\ell $.}
    \label{fig:notationbisector}
\end{figure}

Depending on whether $ i $ is on $ B_\ell $ or $ B_r $, there are two main cases. If $ i \in B_\ell $, then depending on whether $ \delta(u,v) $ intersects $ \ell $ or not there are two subcases. If $ \delta(u,v) $ intersects $ \ell $, then the intersection point between $ \ell $ and $ D $ is a Type (2) distinguished point, see Case L1 in~\cite[Section $ 4.2.1 $]{chen2013computing}. If $ \delta(u,v) $ does not intersect $ \ell $, let $ a $ be, as defined in~\cite{chen2013computing}, the first Voronoi vertex encountered that is incident to $ h $ when we move along the boundary of $ C(D) $ counterclockwise. There are at most two tangent lines from $ a $ to $ D $; then the first tangency point of these two tangent lines we meet when moving from $ u $ along the reversing direction is a Type (3) distinguished point, see Case L2 in~\cite[Section $ 4.2.1 $]{chen2013computing}. If $ i \in B_r $, then there are three subcases depending on whether (i) $ \delta(u,v) $, or (ii) the ray $ \rho(v_1) $ starting at $ v_1 $ and along the direction from $ u_1 $ to $ v_1 $ intersects $ \ell \cup D' $ or not. If $ \delta(u,v) $ intersects $ \ell \cup D' $, then the intersection point between $ \ell $ and $ D $ is a Type~(2) distinguished point, see Case R1 in~\cite[Section $ 4.2.2 $]{chen2013computing}. If $ \rho(v_1) $ intersects $ \ell \cup D' $, let the point $ a $ be, as defined in~\cite{chen2013computing}, the Voronoi vertex incident to $ h $ as going from $ i$ to $ h $ clockwise around $ C(D) $, and let the point $ c $ be the other Voronoi vertex incident to $ h $, then the closest point on $ D $ from $ a $ or the closest point on $ D $ from $ c $ are Type (4) distinguished points, see Case R2 in~\cite[Section $ 4.2.2 $]{chen2013computing}. Finally, if $ \rho(v_1) $ does not intersect $ \ell \cup D' $, then the tangency point of the outer tangent line from $ D' $ to $ D $ that intersects $ B_r $ is a Type (1) distinguished point, see Case R3 in~\cite[Section $ 4.2.2 $]{chen2013computing}. We can prove all these distinguished points are on the arc $ \widehat{ut} $ by following the proof of the five cases in~\cite{chen2013computing} by considering the point $ v_1 $ instead of point $ v $, and the weighted metric for the distances.

This analysis completes the proof of Proposition~\ref{prop:lengthchain}. Since the size of the Voronoi diagram $ \mathit{VD}(\mathfrak D) $ is $ O(n) $, the number of distinguished points obtained from $ \mathit{VD}(\mathfrak D) $ is $ O(n) $. Hence, the coalesced graph $ G_{\mathcal{V}}^c $ has $ O(n) $ vertices and~$ O(k) $ edges, and we obtain the following result.

\begin{corollary}
    A shortest path in $ G_{\mathcal{V}}^c $ can be computed in $ O(n\log{n}+k) $ time.
\end{corollary}

In addition, a shortest path in $ G_{\mathcal{V}}^c $ corresponds to a shortest path in $ G_{\mathcal{V}}^d $ with the same weighted length. This result was proved by Chen et al. in~\cite[Lemma~$ 3.2 $]{chen2013computing} when the space only contains $ \infty $-regions. In our more general case, we can use the same approach taking into account the weighted region metric.

\section{Discretization}
\label{sec:discretization}

In this section we focus on the discretization scheme of the more general case where the disks have associated any non-negative weight. We construct a weighted graph $ G_\varepsilon = (V_\varepsilon(G_\varepsilon), E_\varepsilon(G_\varepsilon)) $ by carefully adding Steiner points on the boundary of the disks. Then, one can apply Dijkstra’s algorithm on $ G_\varepsilon $ to obtain a path $ \tilde{\pi}(s,t) $ that is a $ (1+\varepsilon)$-approximation of $\mathit{SP_w}(s,t)$.

Lemma~\ref{lem:greaterpi2} states that when the weight of a disk is at least $ \frac{\pi}{2} $, then no shortest path will intersect the interior of that disk. In this paper we are discretizing the space to obtain a $ (1+\varepsilon) $-approximation of a shortest path. Thus, from now on, we can assume, without loss of generality, that the maximum weight of the regions is $ \frac{\pi}{2} $.

First, we introduce the value $ d_i $ defined as the minimum Euclidean distance from $ D_i $ to any other disk $ D_j $ in the setting $ \mathfrak D $. We also define a \emph{weighted angular radius} $ \alpha_i $ to be $ \alpha_i = \arcsin{\left(\frac{\min\{d_i,R_i\}\min\{1, \omega_i\}}{4 R_i\max\{1, \omega_i\}}\right)} $. Observe that $ \alpha_i $ is not larger than $ \arcsin{\left(\frac{1}{4}\right)} < \frac{\pi}{2} $, so $ \alpha $ is well defined.

\begin{definition}
    \label{def:vertexvicinity}
    Let $ v $ be a point on the boundary of a disk $ D_i \in \mathfrak{D} $. We refer to the disk with center $ v $ and radius $ 2 R_i\sin{\alpha_i}$ as the \emph{vertex vicinity of vertex $ v $}, or \emph{vertex vicinity} when $ v $ is clear from the context.
\end{definition}

Observe that in Definition~\ref{def:vertexvicinity}, $\alpha_i $ is equal to the minimum angle between (1) the line tangent to $ D_i $ at $ v $ and (2) the line through $ v $ and the intersection point between the vertex vicinity of $ v $ and $ D_i $. See Figure~\ref{fig:anglealpha}.

\begin{figure}[tb]
    \centering
    \includegraphics[width = 0.5\textwidth]{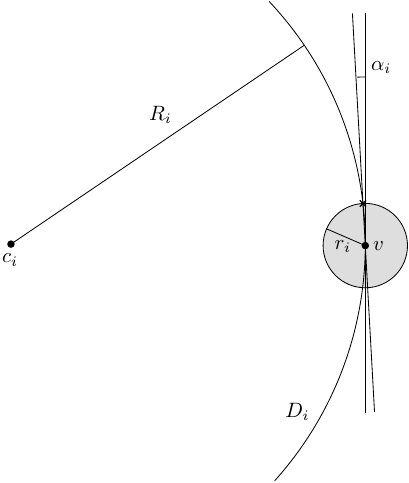}
    \caption{Vertex vicinity of a vertex $ v$ (in grey) on the boundary of disk $ D_i $.}
    \label{fig:anglealpha}
\end{figure}

We use the definition of vertex vicinity to place $ k_i $ points around each disk~$ D_i $. These points $ v_i^1, \ldots, v_i^{k_i} $, called \emph{vertex vicinity centers}, are equally spaced around the disk~$ D_i $.

If the weight of the disk $ D_i $ is $ 0 $, we define the angular distance between the center of the disk $ c_i $ and two consecutive vertex vicinity centers $ v_i^\ell $ and $ v_{i}^{\ell+1} $, for $ 1 \leq \ell < k_i $, by

\begin{equation}
    \label{eq:distanceSteiner0}
    \angle v_i^\ell c_iv_{i}^{\ell+1} = \frac{\varepsilon d_i}{a(d_i+1)}, \text{ for } \varepsilon \in \left(0,\frac{a\pi}{2\omega_i}\right],
\end{equation}
where $ a = \frac{1+3c+\sqrt{9c^2+10c+1}}{2}$, $ c = \frac{\frac{\pi}{2}\max_{1\leq j \leq n}\{R_j\}}{\min_{1\leq j \leq n}\{d_j\}}$, and $ k_i $ is taken so that it is the largest integer satisfying $ \angle v_i^1c_iv_{i}^{k_i} < 2\pi $. Note that we do not consider the particular case where the disks overlap, so $ c $ is a positive value, and $ a > 1$.

Otherwise, if the weight of $ D_i $ is $ \omega_i > 0$, we have that $ k_i = \left\lfloor \frac{\pi}{2\alpha_i} \right\rfloor $. Let $ p_{i,j}^0 $ be the point diametrically opposed to $ v_i^j $ in $ D_i $. We associate to each vertex vicinity center $ v_i^j $ a set of $ 2r+1 $ points on $ D_i $, noted as $ p_{i,j}^0, \ldots, p_{i,j}^{2r} $, each of which is called a \emph{ring point}. The ring points $ p_{i,j}^0, \ldots, p_{i,j}^{r} $ are placed on the boundary of the half disk to the right of $ \overrightarrow{v_{i}^jp_{i,j}^0} $, from $ p_{i,j}^0 $ to $ v_{i}^j $ in clockwise order. The points $ p_{i,j}^{r+1}, \ldots, p_{i,j}^{2r} $ are placed symmetrically on the other half of the disk. We impose the condition that inside the vertex vicinity of each vertex $ v_i^j$, we do not place ring points associated to $ v_i^j $. Hence, $ r $ is taken so that it is the largest integer satisfying $ \angle p_{i,j}^{0}v_i^jp_{i,j}^{r} \leq \frac{\pi}{2} - \alpha_i $. See Figure~\ref{fig:ringpoints} for an illustration of the ring points associated to a vertex vicinity center.

\begin{figure}[tb]
    \centering
    \includegraphics[width = 0.8\textwidth]{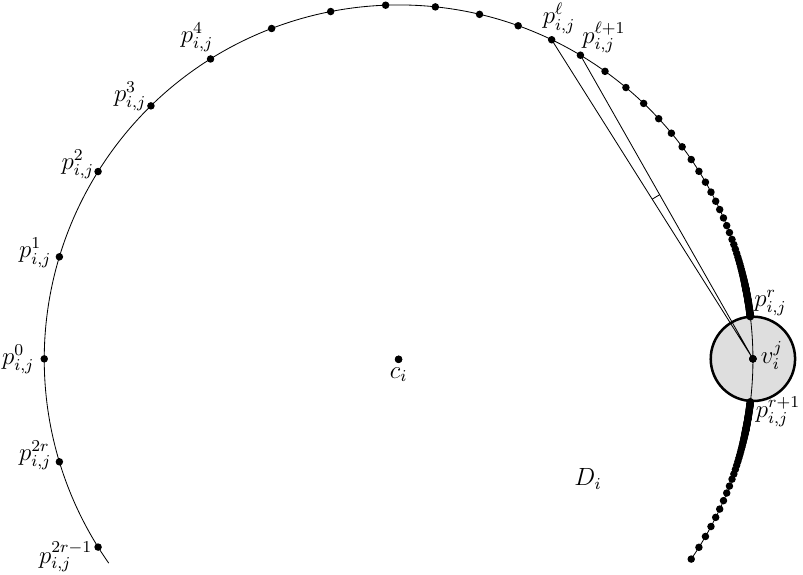}
    \caption{Ring points associated to the vertex vicinity center $ v_i^j $ on the boundary of $ D_i $.}
    \label{fig:ringpoints}
\end{figure}

In addition, we define the angular distance between a vertex vicinity center~$ v_{i}^j $ and two consecutive ring points $ p_{i,j}^\ell $ and $ p_{i,j}^{\ell+1} $, for $ 0 \leq \ell < r $, by

\begin{equation}
    \label{eq:distanceSteiner}
    \angle p_{i,j}^\ell v_i^jp_{i,j}^{\ell+1} = \frac{\omega_i\varepsilon}{a}\left(1-\frac{2\omega_i\varepsilon}{a\pi}\right)^\ell, \text{ for } \varepsilon \in \left(0,\frac{a\pi}{2\omega_i}\right].
\end{equation}

Based on the discretizations that have been used before, and without loss of generality, we may assume that the source vertex $ s $ and the target vertex $ t $ are vertex vicinity centers or ring points. 

From Equation (\ref{eq:distanceSteiner}) and using the definition of vertex vicinity, we can obtain the number of ring points associated to $ v_i^j $ that we are adding to the boundary of a disk $ D_i $ with weight $ \omega_i > 0$. This result can be obtained using the formula for the sum of the first terms of a geometric series.

    \begin{proposition}
        \label{prop:steinerone}
        Let $ D_i $ be a disk with weight $ \omega_i > 0 $. The number of ring points associated to $ v_i^j $ added to the boundary of $ D_i $ is upper-bounded by $ 2\frac{1+\log_2{\frac{\alpha_i}{\pi}}}{\log_2{\left(1-\frac{2\omega_i\varepsilon}{a\pi}\right)}} +1 $.
    \end{proposition}

    \begin{proof}
        By the definition of the ring points we know that the angular distance between consecutive Steiner points on half a disk is $ \angle p_{i,j}^\ell v_i^jp_{i,j}^{\ell+1} = \frac{\omega_i\varepsilon}{a}\left(1-\frac{2\omega_i\varepsilon}{a\pi}\right)^\ell $. So the angular distance from point $ p_{i,j}^0 $ to each of the Steiner points $ p_{i,j}^\ell, \ \ell \in \{1, \ldots, r\}$ is given by
        
        \begin{align*}
	        \angle p_{i,j}^0v_i^jp_{i,j}^\ell & = \sum_{m=0}^{\ell-1} \frac{\omega_i\varepsilon}{a}\left(1-\frac{2\omega_i\varepsilon}{a\pi}\right)^m = \frac{\omega_i\varepsilon}{a} \cdot \frac{1-\left(1-\frac{2\omega_i\varepsilon}{a\pi}\right)^{\ell}}{1-\left(1-\frac{2\omega_i\varepsilon}{a\pi}\right)} \\
         & = \frac{\omega_i\varepsilon}{a} \cdot \frac{1-\left(1-\frac{2\omega_i\varepsilon}{a\pi}\right)^{\ell}}{\frac{2\omega_i\varepsilon}{a\pi}} = \frac{\pi}{2} \left(1-\left(1-\frac{2\omega_i\varepsilon}{a\pi}\right)^{\ell}\right).
	    \end{align*}

        The largest value of $ r $ such that $ \angle p_{i,j}^0v_i^jp_{i,j}^r \leq  \frac{\pi}{2} -\alpha_i$ can be found by solving the following inequality:

        \begin{align*}
            \frac{\pi}{2} \left(1-\left(1-\frac{2\omega_i\varepsilon}{a\pi}\right)^{r}\right) & \leq \frac{\pi}{2} -\alpha_i \Longrightarrow 1-\left(1-\frac{2\omega_i\varepsilon}{a\pi}\right)^{r} \leq 1-\frac{2\alpha_i}{\pi} \\
            & \Longrightarrow \left(1-\frac{2\omega_i\varepsilon}{a\pi}\right)^{r} \geq \frac{2\alpha_i}{\pi}
            \Longrightarrow r \leq \log_{1-\frac{2\omega_i\varepsilon}{a\pi}} \frac{2\alpha_i}{\pi} \\
            & = \frac{\log_2{\frac{2\alpha_i}{\pi}}}{\log_2{\left(1-\frac{2\omega_i\varepsilon}{a\pi}\right)}} = \frac{1+\log_2{\frac{\alpha_i}{\pi}}}{\log_2{\left(1-\frac{2\omega_i\varepsilon}{a\pi}\right)}}.
        \end{align*}

        Since we need to add ring points around the whole disk $ D_i $, we need two times the number of points in the previous equation. In addition, we need to take into account vertex $ p_{i,j}^0 $, hence the final result.
    \end{proof}

    Observe that we do not need to place the number of points from Proposition~\ref{prop:steinerone} around each vertex vicinity center $ v_i^j$, since some of them are further away than the ring points associated to neighboring vertex vicinity centers. Hence, we create an annulus around each point $ v_i^j$ and we place ring points only inside these annuli. The smallest circumference of the annuli, i.e., the boundary of the vertex vicinity of $ v_i^j$, has radius $ 2R_i\sin{\alpha_i}$, and the largest circumference has radius $ 2R_i\sin{(2\alpha_i)}$. Hence, an upper bound on the total number of points placed on the boundary of each disk $ D_i $ with weight $ \omega_i > 0 $ is given next.

    \begin{proposition}
        \label{prop:total}
        Let $ D_i $ be a disk with weight $ \omega_i > 0 $. The total number of vertex vicinity centers and ring points added to the boundary of $ D_i $ is upper-bounded by $ \frac{1}{\log_2{\frac{a\pi}{a\pi-2\omega_i\varepsilon}}} \frac{\pi}{\alpha_i}$.
    \end{proposition}

    \begin{proof}
        Let $ v_i^j$ be a vertex vicinity center on the boundary of a disk $ D_i $, for $ j \in \left\{1, \ldots, \frac{\pi}{2\alpha_i}\right\} $. By Proposition~\ref{prop:steinerone} we place $ 2\frac{1+\log_2{\frac{\alpha_i}{\pi}}}{\log_2{\left(1-\frac{2\omega_i\varepsilon}{a\pi}\right)}}+1 $ points associated to $  v_i^j$ on the boundary of $ D_i $. Since we are creating an annulus where the largest circumference has radius $ 2\cos{\alpha_i} $ times the radius of the small one (see Figure~\ref{fig:ring}), outside this second disk, and around $ D_i $, we are placing $ 2\frac{1+\log_2{\frac{2\alpha_i}{\pi}}}{\log_2{\left(1-\frac{2\omega_i\varepsilon}{a\pi}\right)}}+1 $ points. Note the $ 2 $ that is multiplying the term inside the logarithm in the numerator of the equation.

        \begin{figure}[tb]
            \centering
            \includegraphics[width = 0.5\textwidth]{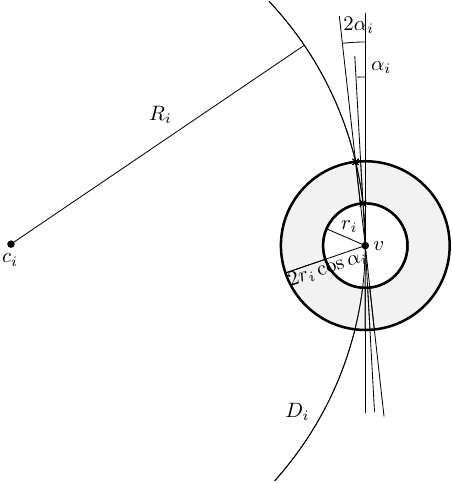}
            \caption{The annulus where the ring points from $ v $ are placed is represented in grey.}
            \label{fig:ring}
        \end{figure}

        Then, the idea is to calculate the number of Steiner points inside the annulus:

        \begin{align*}
            2\left(\frac{1+\log_2{\frac{\alpha_i}{\pi}}}{\log_2{\left(1-\frac{2\omega_i\varepsilon}{a\pi}\right)}} - \frac{1+\log_2{\frac{2\alpha_i}{\pi}}}{\log_2{\left(1-\frac{2\omega_i\varepsilon}{a\pi}\right)}}\right) & = 2\frac{\log_2{\frac{\alpha_i}{\pi}}-\log_2{\frac{2\alpha_i}{\pi}}}{\log_2{\left(1-\frac{2\omega_i\varepsilon}{a\pi}\right)}} = \frac{2\log_2{\frac{\frac{\alpha_i}{\pi}}{\frac{2\alpha_i}{\pi}}}}{\log_2{\left(1-\frac{2\omega_i\varepsilon}{a\pi}\right)}} \\
            & = \frac{2\log_2{\frac{1}{2}}}{\log_2{\left(1-\frac{2\omega_i\varepsilon}{a\pi}\right)}} = \frac{-2}{\log_2{\frac{a\pi-2\omega_i\varepsilon}{a\pi}}} \\
            & = \frac{2}{\log_2{\frac{a\pi}{a\pi-
            2\omega_i\varepsilon}}}.
        \end{align*}
        Note that in the previous equation we are not counting the intersection points between $ D_i $ and the largest circumference around $ v_i^j$. In addition, the intersection points between $ D_i $ and the smallest circumference around $ v_i^j$ coincide with the vertex vicinity centers of $ v_i^{j-1} $ and $ v_i^{j+1}$. Finally, since the vertex vicinity centers belong to the set $ V_\varepsilon(G_\varepsilon) $ of nodes, and we have $ \frac{\pi}{2\alpha_i} $ vertex vicinity centers, we get the desired result.
    \end{proof}

    Next, we provide a result for the total number of Steiner points added to the graph $ G_\varepsilon $, as well as the total number of edges.

    \begin{lemma}
        The number of nodes in $ G_\varepsilon $ is $ O\left(\frac{n}{\varepsilon}\right) $, and the number of edges is $ O\left(\frac{n^2}{\varepsilon^2}\right) $. 
    \end{lemma}

    \begin{proof}

        Proposition~\ref{prop:total} gives us an upper bound on the number of vertex vicinity centers and ring points on each disk with weight greater than $ 0 $, which is $ \frac{1}{\log_2{\frac{a\pi}{a\pi-2\omega_i\varepsilon}}} \frac{\pi}{\alpha_i}$. We know that $ \log_2{\frac{1}{1-x}} \geq x $, when $ x \geq 0 $, so
        
        \begin{align*}
            \frac{1}{\log_2{\frac{a\pi}{a\pi-2\omega_i\varepsilon}}} \frac{\pi}{\alpha_i} & = \frac{\pi}{\alpha_i\log_2{\frac{1}{1-\frac{2\omega_i\varepsilon}{a\pi}}}} \leq \frac{\pi}{\alpha_i\frac{2\omega_i\varepsilon}{a\pi}} \leq \frac{\pi}{\arcsin{\left(\frac{\min\{d_i, R_i\}\min\{1,\omega_i\}}{4R_i\max\{1, \omega_i\}}\right)}\frac{2\omega_i\varepsilon}{a\pi}} \\
            & \leq \frac{\pi}{\left(\frac{\min\{d_i, R_i\}\min\{1,\omega_i\}}{4R_i\max\{1, \omega_i\}}\right)\frac{2\omega_i\varepsilon}{a\pi}}
            = \frac{2a\pi^2R_i\max\{1,\omega_i\}}{\min\{d_i,R_i\}\min\{1,\omega_i\}\omega_i\varepsilon} \\
            & \leq \frac{2a\pi^2(R_i+1)\frac{\pi}{2}}{\min\{d_i,R_i\}\min\{1,\omega_i\}\omega_i\varepsilon} < \frac{2a\pi^3(R_i+1)}{\min\{d_i,R_i\}\min\{\omega_i,\omega_i^2\}\varepsilon} \\
            & \leq \frac{2a\pi^3(R_i+1)}{\min\{d_i,R_i\}\min\{1,\omega_i^2\}\varepsilon}.
        \end{align*}

        Moreover, if the weight of $ D_i $ is $ 0 $, then by Equation (\ref{eq:distanceSteiner0}) we are placing $ \frac{2\pi}{\frac{\varepsilon d_i}{a(d_i+1)}R_i} $ points on the boundary of $ D_i $, so
        
        \begin{align*}
            \frac{2\pi}{\frac{\varepsilon d_i}{a(d_i+1)}R_i} & = \frac{2a\pi(d_i+1)}{R_i\varepsilon d_i} = \begin{cases}
                \overbrace{=}^{\text{if } R_i \leq d_i} \frac{2a\pi(R_i+1)}{R_i^2\varepsilon} \\
                \overbrace{=}^{\text{if } d_i < R_i} \frac{2a\pi(d_i+1)}{d_i^2\varepsilon}
            \end{cases} = \frac{2a\pi(\min\{d_i,R_i\}+1)}{\min\{d_i^2,R_i^2\}\varepsilon} \\
            & \leq \frac{2a\pi(R_i+1)}{\min\{d_i^2,R_i^2\}\varepsilon} < \frac{2a\pi^3(R_i+1)}{\min\{d_i^2,R_i^2\}\varepsilon}.
        \end{align*}
        
        Then, since we have $ n $ disks, the total number of vertex vicinity centers and ring points is upper-bounded by $ C(\mathfrak D)\frac{n}{\varepsilon} $, where $ C(\mathfrak D) < \frac{2a\pi^3(\max_{1 \leq j \leq n}\{R_j\} + 1)}{\min\{1,\omega^2\} \cdot\min_{1 \leq j \leq n}\{1, d_j^2, R_j^2\}} $, where $ \omega $ is the minimum positive weight of $ \mathfrak D $. Thus, the estimate on the number of nodes in $ G_\varepsilon $ is $ O(\frac{n}{\varepsilon})$. Note that here we are taking into account that $ a $ does not depend on $ \varepsilon $.

        The set $ E_\varepsilon $ of edges is obtained by creating an edge $ (u,v) $ between any two vertex vicinity centers or ring points. In case $ u $ and $ v $ are adjacent on the disk, we add an arc between them. In addition, if $ u $ and $ v $ are not visible, the edge is a non-straight line segment. This edge is a shortest path from $ u $ to $ v $ avoiding all the disks, see, e.g., the red path in Figure~\ref{fig:toutside}. Thus, the total number of edges in $ G_\varepsilon $ is $ O\left(\frac{n^2}{\varepsilon^2}\right) $.
    \end{proof}

In case that two nodes are not visible or are adjacent on the boundary of a disk, we join them by arcs of the disks, instead of using a straight-line segment. This way, we ensure that if a shortest path between a pair of points does not intersect the interior of the disks, then our algorithm computes the shortest path exactly.

\section{Discrete path}
\label{sec:onedisk}

If both the source point $ s $ and the target point $ t $ are on the boundary of the same disk $ D $ and the only disk $ \mathit{SP_w}(s,t) $ intersects is $ D $, then we can compute a shortest path $ \mathit{SP_w}(s,t) $ exactly, and in constant time. This result is obtained by taking into account that there are only two possible shortest paths from $ s $ to $ t $, see Observation~\ref{obs:length}.

Now, suppose $ s $ is on the boundary of a disk $ D $ centered at $ c $, and $ t $ is outside~$ D $, see Figure~\ref{fig:toutside}. We prove that there is a path $ \tilde{\pi}(s,t) $ whose length is at most $ \left(1+\frac{\varepsilon}{a}\right) $ times larger than the length of a shortest path from $ s $ to $ t $ when intersecting only the disk $ D $. This path $ \tilde{\pi}(s,t) $ is a shortest path through the vertices of the discretization.

\begin{figure}[tb]
    \centering
    \includegraphics[width = 0.7\textwidth]{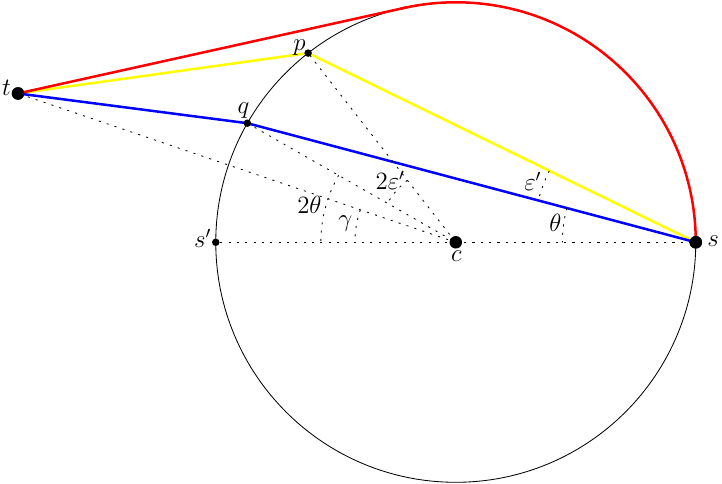}
    \caption{Notation when $ t $ is outside the disk.}
    \label{fig:toutside}
\end{figure}

In fact, we can compute a shortest path exactly in this case. However, the result in Lemma~\ref{lem:toutside} will be useful to prove the approximation ratio when a shortest path intersects more than one disk.

\begin{lemma}
    \label{lem:toutside}
    Let $ s $ be a point on the boundary of a disk $ D $ centered at $ c $ and weight~$ \omega \geq 0 $, and let $ t $ be a point outside $ D $. If $ D $ is the only disk intersected by $ \mathit{SP_w}(s,t) $, then $ \lVert \tilde{\pi}(s,t) \rVert \leq \left(1+\frac{\varepsilon}{a}\right)\cdot \lVert \mathit{SP_w}(s,t) \rVert $.
\end{lemma}

\begin{proof}
    First, we prove the case where $ \omega > 0 $. Suppose that a shortest path from $ s $ to $ t $ does not intersect the interior of $ D $. In this case, the approximate shortest path intersects a shortest arc of the disk from $ s $ to a tangency point from~$ t $ to $ D $. Before this tangency point, and along the boundary of $ D $, there is a ring point. This ring point is joined to $ t $ by an edge which is not a straight-line segment. This means that in this case, an approximate shortest path is also a shortest path. Hence, $ \lVert \tilde{\pi}(s,t) \rVert = \lVert \mathit{SP_w}(s,t) \rVert \leq \left(1+\frac{\varepsilon}{a}\right) \lVert \mathit{SP_w}(s,t) \rVert $.
    
    Now, suppose that $ \mathit{SP_w}(s,t) $ intersects the interior of $ D $. Let $ q $ be the point where $ \mathit{SP_w}(s,t) $ leaves the disk and let $ \theta $ be the angle $ \angle{csq} $. Let $ p $ be the closest ring point to $ q $ on the boundary of the disk, and let $ \angle{csp} $ be $ \theta+\varepsilon'$, for some $ \varepsilon' \geq 0 $ and $ \theta < \frac{\pi}{2}$. The case where $ \varepsilon' < 0 $ will be addressed later in the proof. In addition, we can assume, without loss of generality, that $ D $ has radius length~$ 1 $. Let $ s' $ be the point diametrically opposed to $ s $ on the boundary of $ D $, and let $ \gamma $ be the angle $ \angle{s'ct} $, see Figure~\ref{fig:toutside}. Then, the length of the approximate path is
    
    \begin{equation*}
        \lVert \tilde{\pi}(s,t) \rVert = 2\omega\cos{(\theta+\varepsilon')}+\sqrt{1+|ct|^2-2|ct|\cos{(2(\theta+\varepsilon')-\gamma)}},
    \end{equation*}

    and the length of $ \mathit{SP_w}(s,t) $ is

    \begin{equation*}
        \lVert \mathit{SP_w}(s,t) \rVert = 2\omega\cos{\theta}+\sqrt{1+|ct|^2-2|ct|\cos{(2\theta-\gamma)}}.
    \end{equation*}    

    Thus, we would like to prove that:

    \begin{equation*}
        \frac{2\omega\cos{(\theta+\varepsilon')}+\sqrt{1+|ct|^2-2|ct|\cos{(2(\theta+\varepsilon')-\gamma)}}}{2\omega\cos{\theta}+\sqrt{1+|ct|^2-2|ct|\cos{(2\theta-\gamma)}}} \leq \frac{\omega\cos{\theta}+\varepsilon'}{\omega\cos{\theta}}.
    \end{equation*}

    Since $ \cos{(\theta+\varepsilon')} \leq \cos{\theta} $, it is sufficient to prove that:
    \[\frac{2\omega\cos{\theta}+\sqrt{1+|ct|^2-2|ct|\cos{(2(\theta+\varepsilon')-\gamma)}}}{2\omega\cos{\theta}+\sqrt{1+|ct|^2-2|ct|\cos{(2\theta-\gamma)}}} \leq \frac{\omega\cos{\theta}+\varepsilon'}{\omega\cos{\theta}}.\]

    Thus, 
    
    \begin{multline*}      2\omega^2\cos^2{\theta}+\omega\cos{\theta}\sqrt{1+|ct|^2-2|ct|\cos{(2(\theta+\varepsilon')-\gamma)}} \\
    \leq 2\omega^2\cos^2{\theta}+2\varepsilon'\omega\cos{\theta}
    + \sqrt{1+|ct|^2-2|ct|\cos{(2\theta-\gamma)}}(\omega\cos{\theta}+\varepsilon')
    \end{multline*}
    \begin{multline*}
        \omega\cos{\theta}\sqrt{1+|ct|^2-2|ct|\cos{(2(\theta+\varepsilon')-\gamma)}} \\
        \leq \ 2\varepsilon'\omega\cos{\theta} + \sqrt{1+|ct|^2-2|ct|\cos{(2\theta-\gamma)}}(\omega\cos{\theta}+\varepsilon')
        \end{multline*}
        \begin{multline*}
        \left(\sqrt{1+|ct|^2-2|ct|\cos{(2(\theta+\varepsilon')-\gamma)}}-2\varepsilon'-\sqrt{1+|ct|^2-2|ct|\cos{(2\theta-\gamma)}}\right)\omega\cos{\theta} \\
        \leq \ \sqrt{1+|ct|^2-2|ct|\cos{(2\theta-\gamma)}}\varepsilon'.
    \end{multline*}

    The term to the right of the last inequality represents the length of the shortest path from $ s $ to $ t $ outside the disk, multiplied by a positive value $ \varepsilon' $, so this value is positive. In addition, $ \omega > 0 $, and since $ \theta < \frac{\pi}{2}, \ \cos{\theta} > 0 $. Then, it is enough to prove that $ \sqrt{1+|ct|^2-2|ct|\cos{(2(\theta+\varepsilon')-\gamma)}}-2\varepsilon'-\sqrt{1+|ct|^2-2|ct|\cos{(2\theta-\gamma)}} \leq 0 $. Thus,
    
    \begin{align*}
        \sqrt{1+|ct|^2-2|ct|\cos{(2(\theta+\varepsilon')-\gamma)}} \leq & \ 2\varepsilon'+\sqrt{1+|ct|^2-2|ct|\cos{(2\theta-\gamma)}} \\
        1+|ct|^2-2|ct|\cos{(2(\theta+\varepsilon')-\gamma)} \leq & \ 4\varepsilon'^2+1+|ct|^2-2|ct|\cos{(2\theta-\gamma)} \\
        & \ +4\varepsilon'\sqrt{1+|ct|^2-2|ct|\cos{(2\theta-\gamma)}} \\ -2|ct|\cos{(2(\theta+\varepsilon')-\gamma)} \leq & \ 4\varepsilon'^2-2|ct|\cos{(2\theta-\gamma)}\\
        & \ +4\varepsilon'\sqrt{1+|ct|^2-2|ct|\cos{(2\theta-\gamma)}} \\
        |ct|(\cos{(2\theta-\gamma)}-\cos{(2(\theta+\varepsilon')-\gamma)}) \leq & \ 2\varepsilon'\left(\varepsilon'+\sqrt{1+|ct|^2-2|ct|\cos{(2\theta-\gamma)}}\right).
    \end{align*}

    Since $ \cos{a}-\cos{b} = -2\sin{\frac{a+b}{2}}\sin{\frac{a-b}{2}} $, for all angles $ a,b $, the previous inequality is equivalent to:

    \begin{align*}
        -|ct|\sin{\left(\frac{4\theta-2\gamma+2\varepsilon'}{2}\right)}\sin{(-\varepsilon')} & \leq \varepsilon'^2+\varepsilon'\sqrt{1+|ct|^2-2|ct|\cos{(2\theta-\gamma)}} \\
        |ct|\sin{\left(\frac{4\theta-2\gamma+2\varepsilon'}{2}\right)}\sin{\varepsilon'} & \leq \varepsilon'^2+\varepsilon'\sqrt{1+|ct|^2-2|ct|\cos{(2\theta-\gamma)}}.
    \end{align*}

    We know that $ \sin{a} \leq a $, for $ a \geq 0 $, so it is sufficient to prove that:
    
    \begin{align*}
        |ct|\sin{(2\theta-\gamma+\varepsilon')}\varepsilon' & \leq \varepsilon'^2+\varepsilon'\sqrt{1+|ct|^2-2|ct|\cos{(2\theta-\gamma)}} \\
        |ct|\sin{(2\theta-\gamma+\varepsilon')} & \leq \varepsilon'+\sqrt{1+|ct|^2-2|ct|\cos{(2\theta-\gamma)}} \\
        |ct|\sin{(2\theta-\gamma+\varepsilon')} - \varepsilon' & \leq \sqrt{1+|ct|^2-2|ct|\cos{(2\theta-\gamma)}}.
    \end{align*}

    Now, $ \varepsilon' $ is just on the left-hand side of the inequality, so we would like to know what is the largest value $ |ct|\sin{(2\theta-\gamma+\varepsilon')} - \varepsilon' $ can take
    
    \begin{align*}
        \frac{\partial \left(|ct|\sin{(2\theta-\gamma+\varepsilon')} - \varepsilon'\right)}{\partial \varepsilon'} & = |ct|\cos{(2\theta-\gamma+\varepsilon')}-1 = 0 \\
        \Longleftrightarrow 1 & = |ct|\cos{(2\theta-\gamma+\varepsilon')} \\ \Longleftrightarrow \varepsilon' & = \arccos{\left(\frac{1}{|ct|}\right)}+\gamma-2\theta.
    \end{align*}

    We know that $ \varepsilon' \geq 0 $, so the maximum value $ |ct|\sin{(2\theta-\gamma+\varepsilon')} - \varepsilon' $ can take is obtained when $ \varepsilon' = \max\left\{0, \arccos{\left(\frac{1}{|ct|}\right)}+\gamma-2\theta\right\}$. Hence,
    
    \begin{itemize}
        \item If $ \varepsilon'=0 $, it is sufficient to prove that
        
        \begin{align*}
            |ct|\sin{(2\theta-\gamma)} & \leq \sqrt{1+|ct|^2-2|ct|\cos{(2\theta-\gamma)}} \\
            |ct|^2\sin^2{(2\theta-\gamma)} & \leq 1+|ct|^2-2|ct|\cos{(2\theta-\gamma)}\\
            |ct|^2\sin^2{(2\theta-\gamma)} -|ct|^2 & \leq 1-2|ct|\cos{(2\theta-\gamma)} \\
            -|ct|^2\cos^2{(2\theta-\gamma)} & \leq 1-2|ct|\cos{(2\theta-\gamma)}\\
            0 & \leq 1-2|ct|\cos{(2\theta-\gamma)}+|ct|^2\cos^2{(2\theta-\gamma)} \\
            0 &\leq (1-|ct|\cos{(2\theta-\gamma)})^2.
        \end{align*}
        \item If $ \varepsilon' = \arccos{\left(\frac{1}{|ct|}\right)}+\gamma-2\theta $, it is sufficient to prove that $ |ct|\sin{\left(\arccos{\left(\frac{1}{|ct|}\right)}\right)} - \arccos{\left(\frac{1}{|ct|}\right)} - \gamma +2\theta \leq \sqrt{1+|ct|^2-2|ct|\cos{(2\theta-\gamma)}} $:
        
        \begin{align}
            \label{eq:inequality}
            |ct|\sqrt{1-\frac{1}{|ct|^2}} - \arccos{\left(\frac{1}{|ct|}\right)} - \gamma +2\theta & \leq \sqrt{1+|ct|^2-2|ct|\cos{(2\theta-\gamma)}} \nonumber \\
            \sqrt{|ct|^2-1} - \arccos{\left(\frac{1}{|ct|}\right)} - \gamma +2\theta & \leq \sqrt{1+|ct|^2-2|ct|\cos{(2\theta-\gamma)}} \nonumber \\
            2\theta -  \sqrt{1+|ct|^2-2|ct|\cos{(2\theta-\gamma)}} & \leq \arccos{\left(\frac{1}{|ct|}\right)} + \gamma - \sqrt{|ct|^2-1}.
        \end{align}

        Now, $ \theta $ is just on the left-hand side of the last inequality, so we would like to know which is the largest value $ 2\theta -  \sqrt{1+|ct|^2-2|ct|\cos{(2\theta-\gamma)}} $ can take:
        
        \begin{align*}
            \frac{\partial \left(2\theta -  \sqrt{1+|ct|^2-2|ct|\cos{(2\theta-\gamma)}}\right)}{\partial \theta} & = 2-\frac{2|ct|\sin{(2\theta-\gamma)}}{\sqrt{1+|ct|^2-2|ct|\cos{(2\theta-\gamma)}}} = 0 \\
            \sqrt{1+|ct|^2-2|ct|\cos{(2\theta-\gamma)}} & = |ct|\sin{(2\theta-\gamma)} \\
            1+|ct|^2-2|ct|\cos{(2\theta-\gamma)} & = |ct|^2\sin^2{(2\theta-\gamma)} \\
            1 & = |ct|\cos{(2\theta-\gamma)}.
    \end{align*}
    
        Hence, $ 2\theta - \sqrt{1+|ct|^2-2|ct|\cos{(2\theta-\gamma)}} \leq \arccos{\left(\frac{1}{|ct|}\right)}+\gamma-\sqrt{1+|ct|^2-2} = \arccos{\left(\frac{1}{|ct|}\right)}+\gamma-\sqrt{|ct|^2-1} $, which is (\ref{eq:inequality}), and that is what we wanted to prove.
    \end{itemize}

    In both cases, we proved that $ \lVert \tilde{\pi}(s,t) \rVert \leq \frac{\omega\cos{\theta}+\varepsilon'}{\omega\cos{\theta}}\lVert \mathit{SP_w}(s,t) \rVert = (1+\frac{\varepsilon'}{\omega\cos{\theta}}) \lVert \mathit{SP_w}(s,t) \rVert $. We also know that $ \cos{\theta} \geq 1 - \frac{2}{\pi}\theta$ when $ \theta \leq \frac{\pi}{2}$. Hence,

    \[ 1+\frac{\varepsilon'}{\omega\cdot\cos{\theta}} \leq 1+\frac{\varepsilon'}{\omega\left(1-\frac{2}{\pi}\theta\right)} = 1 +\frac{\pi\varepsilon'}{\omega(\pi-2\theta)}. \]
    
    However, we are interested in obtaining a $ \left(1+ \frac{\varepsilon}{a}\right) $-approximation. We can obtain this approximation factor by setting $ \varepsilon' $ to $ \frac{\varepsilon \omega(\pi-2\theta)}{a\pi}$. Note that $ \varepsilon'$ represents the maximum angular distance between consecutive ring points. So, if we place the first Steiner point on the boundary of the disk, diametrically opposed to $ s $, $ \theta=0 $ and, in order to get a $ \left(1+ \frac{\varepsilon}{a}\right) $-approximation, we would like to have the following Steiner point at a distance $ \frac{\varepsilon \omega}{a} $ from the first Steiner point, which is true by Equation (\ref{eq:distanceSteiner}). Following this procedure, one can prove that we always obtain a $ \left(1+ \frac{\varepsilon}{a}\right) $-approximation.

    We also need to calculate the ratio when the closest ring point to $ q $ is to the left of $ q $ with respect to $ \mathit{SP_w}(s,t)$ when oriented from $ s $ to $ t $. Suppose $ p'$ is the closest ring point to $ q $ on the boundary of the disk where $ \angle csp' = \theta+\varepsilon'$ and $ \varepsilon' < 0 $. If the approximation path through $ p $ is shorter than through $ p' $, the algorithm that calculates the approximation path will never go through $ p' $, so we do not need to calculate the ratio of the approximation path through $ p'$. Otherwise, since the length of the approximation path is on the numerator, and the length of the shortest path does not change, the ratio when taking the approximate path through $ p $ is larger. Also, note that in this case, if $ p $ is before the tangency point $ q' $ from $ t $ to~$ D $, the segment from $ p $ to $ t $ will intersect the interior of the disk, see Figure~\ref{fig:tangencyafter}. Hence, we need to calculate the ratio in this particular case. 
    
    \begin{figure}[tb]
    \centering
    \includegraphics[width = 0.6\textwidth]{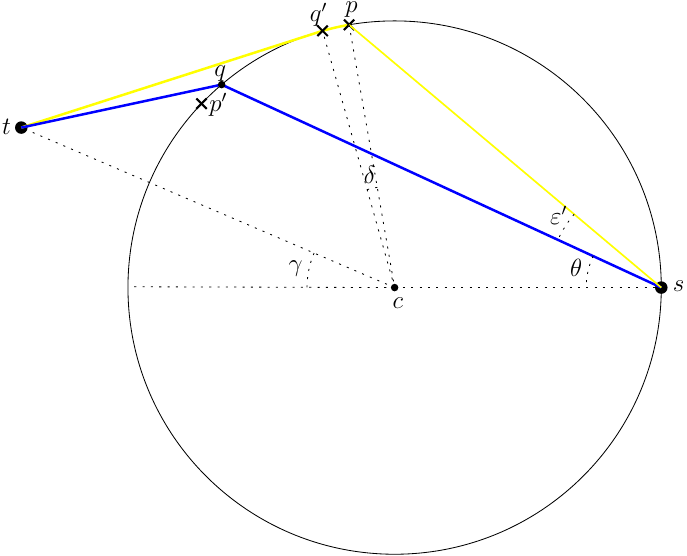}
    \caption{The tangency point from $ t $ to the disk $ D $ is after the Steiner point.}
    \label{fig:tangencyafter}
\end{figure}
    
    The length of the approximate path is given by $ \lVert \tilde{\pi}(s,t) \rVert = \omega|\overline{sp}| + \delta + |\overline{q't}| $, where $ \delta = \angle{pcq'} $, see the yellow path in Figure~\ref{fig:tangencyafter}. We know from before that $ \omega|\overline{sp}| = 2\omega\cos{(\theta+\varepsilon')} $. We also know that $ |\overline{q't}| = \sqrt{|ct|^2-1} $, since $ q' $ is a tangency point from $ t $ to $ D $. In addition, by using some trigonometric identities, we obtain that $ \delta = 2(\theta+\varepsilon')-\arcsin{\left(\frac{\sqrt{\lvert ct \rvert^2-1}}{\lvert ct \rvert}\right)} - \gamma $. Thus,
    
    \begin{equation*}
        \lVert \tilde{\pi}(s,t) \rVert = 2\omega\cos{(\theta+\varepsilon')}+2(\theta+\varepsilon')-\arcsin{\left(\frac{\sqrt{\lvert ct \rvert^2-1}}{\lvert ct \rvert}\right)} - \gamma + \sqrt{|ct|^2-1}.
    \end{equation*}

    Our goal is to prove that:
    
    \begin{equation*}
        \frac{2\omega\cos{(\theta+\varepsilon')}+2(\theta+\varepsilon')-\arcsin{\left(\frac{\sqrt{\lvert ct \rvert^2-1}}{\lvert ct \rvert}\right)} - \gamma + \sqrt{|ct|^2-1}}{2\omega\cos{\theta}+\sqrt{1+|ct|^2-2|ct|\cos{(2\theta-\gamma)}}} \leq \frac{\omega\cos{\theta}+\varepsilon'}{\omega\cos{\theta}}.
    \end{equation*}

    Since $ \cos{(\theta+\varepsilon')} \leq \cos{\theta} $, it is sufficient to prove that:
    
    \begin{equation*} \frac{2\omega\cos{(\theta)}+2(\theta+\varepsilon')-\arcsin{\left(\frac{\sqrt{\lvert ct \rvert^2-1}}{\lvert ct \rvert}\right)} - \gamma + \sqrt{|ct|^2-1}}{2\omega\cos{\theta}+\sqrt{1+|ct|^2-2|ct|\cos{(2\theta-\gamma)}}} \leq \frac{\omega\cos{\theta}+\varepsilon'}{\omega\cos{\theta}}.
    \end{equation*}
    
    Hence,
    
    \begin{align*}
        & \ \Bigg(2\omega\cos{\theta}+2(\theta+\varepsilon')-\gamma-\arcsin{\left(\frac{\sqrt{\lvert ct \rvert^2-1}}{\lvert ct \rvert}\right)}+\sqrt{\lvert ct \rvert^2-1} - 2\omega\cos{\theta} -2\varepsilon' \\
        \ & -\sqrt{1+|ct|^2-2|ct|\cos{(2\theta-\gamma)}}\Bigg)\omega\cos{\theta} \leq \sqrt{1+|ct|^2-2|ct|\cos{(2\theta-\gamma)}}\varepsilon'.
    \end{align*}

    We know that the term on the right-hand side of this inequality is positive, also $ \omega > 0 $, and since $ \theta < \frac{\pi}{2}, \ \cos{\theta} > 0 $. Then, it is sufficient to prove that 
    
    \begin{equation*}
        2\theta-\gamma-\arcsin{\left(\frac{\sqrt{\lvert ct \rvert^2-1}}{\lvert ct \rvert}\right)}+\sqrt{\lvert ct \rvert^2-1} -\sqrt{1+|ct|^2-2|ct|\cos{(2\theta-\gamma)}} \leq 0.
    \end{equation*}

    We now maximize the left-hand side term of the previous inequality with respect to $ \theta $:
    
    \begin{multline*}
        \frac{\partial \left( 2\theta-\gamma-\arcsin{\left(\frac{\sqrt{\lvert ct \rvert^2-1}}{\lvert ct \rvert}\right)}+\sqrt{\lvert ct \rvert^2-1} -\sqrt{1+|ct|^2-2|ct|\cos{(2\theta-\gamma)}}\right)}{\partial \theta} \\
        = 2- \frac{4\lvert ct \rvert\sin{(2\theta-\gamma)}}{2\sqrt{1+\lvert ct \rvert^2-2\lvert ct \rvert \cos{(2\theta-\gamma)}}}
    \end{multline*}

    Now,
    
    \begin{align*}
    0 & = 2- \frac{4\lvert ct \rvert\sin{(2\theta-\gamma)}}{2\sqrt{1+\lvert ct \rvert^2-2\lvert ct \rvert \cos{(2\theta-\gamma)}}} \\
    |ct|\sin{(2\theta-\gamma)} & = \sqrt{1+|ct|^2-2|ct|\cos{(2\theta-\gamma)}} \\
        |ct|^2\sin^2{(2\theta-\gamma)} & = |ct|^2+1-2|ct|\cos{(2\theta-\gamma)} \\
        -|ct|^2\cos^2{(2\theta-\gamma)} & = 1 - 2|ct|\cos{(2\theta-\gamma)} \\
        0 & = \left(1-|ct|\cos{(2\theta-\gamma)}\right)^2 \\
        \theta & = \frac{\arcsin{\left(\frac{\sqrt{|ct|^2-1}}{|ct|}\right)}+\gamma}{2}.
    \end{align*}

    Thus,
    
    \begin{align*}
        2\theta-\gamma-\arcsin{\left(\frac{\sqrt{\lvert ct \rvert^2-1}}{\lvert ct \rvert}\right)}+&\sqrt{\lvert ct \rvert^2-1} -\sqrt{1+|ct|^2-2|ct|\cos{(2\theta-\gamma)}} \\
        \leq & \ \arcsin{\left(\frac{\sqrt{\lvert ct \rvert^2-1}}{\lvert ct \rvert}\right)} - \arcsin{\left(\frac{\sqrt{\lvert ct \rvert^2-1}}{\lvert ct \rvert}\right)} \\
        & + \sqrt{|ct|^2-1}-\sqrt{1+|ct|^2-2|ct|\frac{1}{|ct|}} = 0.
    \end{align*}
    
    This proves that $ \lVert \tilde{\pi}(s,t) \rVert \leq \left(1+\frac{\varepsilon}{a}\right) \lVert \mathit{SP_w}(s,t) \rVert $ when $ s $ is on the boundary of a disk $ D $ with positive weight. 
    
    In the special case where $ \omega = 0 $, we know the exact weight of the shortest path since $ \pi -2\theta + \gamma = \pi \Longrightarrow 2\theta = \gamma $. Hence, $ \lVert \tilde{\pi}(s,t) \rVert = \sqrt{1+|ct|^2-2|ct|\cos{2\varepsilon'}} $, and $ \lVert \mathit{SP_w}(s,t) \rVert = |ct|-1$. Thus,
    
    \begin{align*}
        \frac{\lVert \tilde{\pi}(s,t) \rVert}{\lVert \mathit{SP_w}(s,t) \rVert} & = \frac{\sqrt{1+|ct|^2-2|ct|\cos{2\varepsilon'}}}{|ct|-1} \leq \frac{\sqrt{1+|ct|^2-2|ct|+4|ct|\varepsilon'^2}}{|ct|-1} \\
        & = \sqrt{\frac{(|ct|-1)^2+4|ct|\varepsilon'^2}{(|ct|-1)^2}}= \sqrt{1+\frac{4|ct|\varepsilon'^2}{(|ct|-1)^2}}.
    \end{align*}

    Now, we would like to prove that $ \sqrt{1+\frac{4|ct|\varepsilon'^2}{(|ct|-1)^2}} \leq 1+\frac{\varepsilon}{a} \Longrightarrow 1+\frac{4|ct|\varepsilon'^2}{(|ct|-1)^2} \leq \left(1+\frac{\varepsilon}{a}\right)^2 $. We know that $ |ct| \geq d + 1 $, and $ \varepsilon' \leq \frac{\varepsilon d}{2a(d+1)}$, where $ d $ is the minimum distance from $ D $ to any other disk. Thus, since $ \frac{4|ct|\varepsilon'^2}{(|ct|-1)^2} $ is a decreasing function for $ |ct| > 1 $, we have that
    
    \begin{align*}
        1+\frac{4|ct|\varepsilon'^2}{(|ct|-1)^2} & \leq 1+\frac{4(d+1)\varepsilon'^2}{d^2} \leq 1+\frac{4(d+1)\left(\frac{\varepsilon d}{2a(d+1)}\right)^2}{d^2} = 1 + \frac{\varepsilon^2}{a^2(d+1)} \\
        & \leq 1 + \frac{\varepsilon^2}{a^2} \leq  1 + \frac{\varepsilon^2}{a^2} + \frac{2\varepsilon}{a} = \left(1+\frac{\varepsilon}{a}\right)^2.\qedhere
    \end{align*}
\end{proof}

Next, we generalize the result obtained in Lemma~\ref{lem:toutside} to the case where $\mathit{SP_w}(s,t)$ intersects an ordered sequence of disks $ \mathcal D = (D_j, \ldots, D_k) $. Recall that each disk~$ D_i $ is centered at $ c_i $, has radius $ R_i $, and the weight inside the disk is $ 0 \leq \omega_i \leq \frac{\pi}{2} $, for $ 1 \leq j \leq i \leq k \leq n $. Now recall that $ G_\varepsilon $ is the graph whose vertex set is the set of vertex vicinity centers and ring points, and each pair of points is joined by an edge, see Section~\ref{sec:discretization}.

        \begin{theorem}
            Let $\mathit{SP_w}(s,t)$ be a weighted shortest path between two different points $ s $ and $ t $. There exists a path $ \tilde{\pi}(s,t) $ in $ G_\varepsilon $ such that $ \lVert \tilde{\pi}(s,t) \rVert \leq (1 + \varepsilon)\cdot \Vert \mathit{SP_w}(s, t) \rVert $.
        \end{theorem}

        \begin{proof}
            Let $ \mathcal D = (D_j, \ldots, D_k) $ be the ordered sequence of disks intersected by $ \mathit{SP_w}(s,t)$. We can suppose that $ s \in D_j $, and $ t \in D_k $. Thus, the ordered sequence of points where $ \mathit{SP_w}(s,t)$ first intersects the disks in $ \mathcal D $ is given by $ (s=a_1, a_{2}, \ldots, a_{k-j+1}=t) $, see Figure~\ref{fig:main_theorem}. The portions $ \mathit{SP_w}(a_i, a_{i+1}) $ are called \emph{inter-vertex vicinity portions}.

            \begin{figure}[tb]
                \centering
                \includegraphics[width = \textwidth]{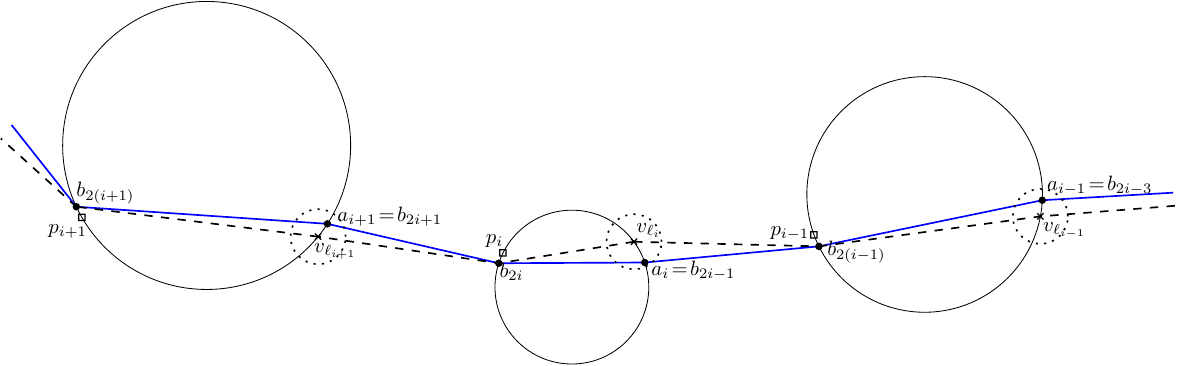}
                \caption{The shortest path $ \mathit{SP_w}(s,t) $ is represented in blue. The path $ \pi'(s,t)$ is represented as a dashed path. The vertex vicinities are the small disks around $ v_{\ell_{i-1}} $, $ v_{\ell_i} $ and $ v_{\ell_{i+1}} $.}
                \label{fig:main_theorem}
            \end{figure}

            We stated in Observation~\ref{obs:length} that there are only two possible ways of intersecting the disks. The subpaths $ \mathit{SP_w}(a_i, a_{i+1}) $ either intersect the interior of the disk $ D_i $, or coincide with an arc of $ D_i $.
	    
            Nodes $ s=a_1 $ and $ t=a_{k-j+1} $ are vertices of $ G_\varepsilon $, so we let $ v_{\ell_1} = s $ and $ v_{\ell_{k-j+1}} = t $. For the remaining points $ a_i, \ 2 \leq i \leq k-j $, we let $ v_{\ell_i} $ be the closest vertex vicinity center or ring point to $ a_i $ in disk $ D_i $, see Figure~\ref{fig:main_theorem}. Consider now an inter-vertex vicinity portion $ \mathit{SP_w}(a_i, a_{i+1}) $. We define the path $ \pi'(v_{\ell_i},v_{\ell_{i+1}}) $ as the path from $ v_{\ell_i} $ to $ v_{\ell_{i+1}} $ through $ b_{2i}$, the point where $ \pi(a_i, a_{i+1}) $ leaves disk~$ D_i $. Using the triangle inequality, the fact that $ a_i $ and $ a_{i+1} $ belong to the interior of the annulus of some vertex vicinity centers, and Equation (\ref{eq:distanceSteiner}) that gives us the distance between adjacent ring points on the same disk, we get that
            
	    \begin{align}
                \label{eq:maxdistance}
	        \lVert \pi'(v_{\ell_i}, v_{\ell_{i+1}}) \rVert \leq & \ \lVert v_{\ell_i}a_{i} \rVert + \lVert \pi(a_i, a_{i+1}) \rVert + \lVert a_{i+1}v_{\ell_{i+1}} \rVert \nonumber \\
                \leq & \ 2R_i\sin{\left(\frac{\omega_i\varepsilon}{2a}\left(1-\frac{2\omega_i\varepsilon}{a\pi}\right)^{\frac{1}{\log_2{\frac{
                a\pi}{a\pi-2\varepsilon\omega_i}}}}\right)}\min\{1,\omega_i\}+ \lVert \pi(a_i, a_{i+1}) \rVert \nonumber \\
                & + 2R_{i+1}\sin{\left(\frac{\omega_{i+1}\varepsilon}{2a}\left(1-\frac{2\omega_{i+1}\varepsilon}{a\pi}\right)^{\frac{1}{\log_2{\frac{
                a\pi}{a\pi-2\varepsilon\omega_{i+1}}}}}\right)}\min\{1,\omega_{i+1}\} \nonumber \\
                \leq & \ \frac{R_i\omega_i\varepsilon}{a} + \lVert \pi(a_i, a_{i+1}) \rVert + \frac{R_{i+1}\omega_{i+1}\varepsilon}{a}.
	    \end{align}

        Note that the inequality we obtain is also true even in the case where $ \omega_i = 0 $ (resp., $ \omega_{i+1}=0 $), since the weight of traveling along the boundary of $ D_i $ (resp., $ D_{i+1} $) is $ 0 $, i.e., $\lVert v_{\ell_i}a_{i} \rVert = 0 \leq \frac{R_i\omega_i\varepsilon}{a}$ (resp., $\lVert a_{i+1}v_{\ell_{i+1}} \rVert = 0 \leq \frac{R_{i+1}\omega_{i+1}\varepsilon}{a}$). Moreover, the maximum distance between consecutive Steiner points on the same disk is given by the last two points on the same annulus. This, together with Equation~(\ref{eq:distanceSteiner}) gives us the second inequality in (\ref{eq:maxdistance}). In the last inequality, we are using the fact that $ \sin{\theta} \leq \theta $ when $ \theta \geq 0 $, and $ (1-x)^y \leq 1 $ when $ x, y \geq 0 $. Also, recall that $ a = \frac{1+3c+\sqrt{9c^2+10c+1}}{2} > 1 $ since $ c= \frac{\frac{\pi}{2}\max_{1\leq j \leq n}\{R_j\}}{\min_{1\leq j \leq n}\{d_j\}} > 0$. In addition, $ a $ can also be written as the solution of the system of equations given by $ a = \frac{bc}{2}$ and $ b = \frac{6a+2}{a-1} $. Then,
        
        \begin{align}
            \label{eq:bound}
            & \ \frac{R_i\omega_i\varepsilon}{a} + \lVert \pi(a_i, a_{i+1}) \rVert + \frac{R_{i+1}\omega_{i+1}\varepsilon}{a} \nonumber\\
            = & \ \frac{2\omega_{i}\varepsilon R_i\min_{1 \leq j \leq n}\{d_j\}}{b\frac{\pi}{2}\max_{1\leq j \leq n}\{R_j\}} + \lVert \pi(a_i, a_{i+1}) \rVert + \frac{2\omega_{i+1}\varepsilon R_{i+1}\min_{1 \leq j \leq n}\{d_j\}}{b\frac{\pi}{2}\max_{1\leq j \leq n}\{R_j\}} \nonumber\\
            \leq & \ \frac{2\omega_{i}\varepsilon R_id_i}{b\omega_iR_i} + \lVert \pi(a_i, a_{i+1}) \rVert + \frac{2\omega_{i+1}\varepsilon R_{i+1}d_{i+1}}{b\omega_{i+1}R_{i+1}} \nonumber\\
            = & \ \frac{2\varepsilon d_i}{b} + \lVert \pi(a_i, a_{i+1}) \rVert + \frac{2\varepsilon d_{i+1}}{b}.
        \end{align}
     
	    Therefore, we obtain the path $ \pi'(s,t) = \pi'(s,v_{\ell_2}) \cup \pi'(v_{\ell_2},v_{\ell_3}) \cup \ldots \cup \pi'(v_{\ell_{k-j}},t) $. 
	    For each $ i = 1, \ldots, k-j $, we define the point $ p_i $ to be the closest Steiner point to $ b_{2i} $ that is to the right of $ \mathit{SP_w}(s, t) $ when oriented from $ s $ to $ t $ if $ b_{2i} $ is to the right of the segment from $ a_{i} $ to its diametrically opposed point on $ D_i $. Otherwise, we let $ p_i $ be the closest Steiner point to $ b_{2i} $ that is to the left of $ \mathit{SP_w}(s, t) $.
     
        Now, we create the path $ \pi''(s,t) = \pi''(s,v_{\ell_2}) \cup \pi''(v_{\ell_2},v_{\ell_3}) \cup \ldots \cup \pi''(v_{\ell_{k-j}},t) $, where $ \pi''(v_{\ell_i}, v_{\ell_{i+1}}) = (v_{\ell_i}, p_{i}, v_{\ell_{i+1}}) $. We know from Observation~\ref{obs:length} and Lemma~\ref{lem:toutside} that $ \lVert \pi''(v_{\ell_i}, v_{\ell_{i+1}}) \rVert \leq \left(1+\frac{\varepsilon}{a}\right)\cdot\lVert \pi'(v_{\ell_i}, v_{\ell_{i+1}}) \rVert $. Thus, using Equation (\ref{eq:bound}),
        
	    \begin{align}
	        \lVert \pi''(s,t) \rVert & = \sum_{i=1}^{k-j} \lVert \pi''(v_{\ell_i}, v_{\ell_{i+1}})\rVert \nonumber \\
                & \leq \left(1+\frac{\varepsilon}{a}\right) \sum_{i=1}^{k-j} \lVert \pi'(v_{\ell_i}, v_{\ell_{i+1}}) \rVert \nonumber \\
	        & \leq \left(1+\frac{\varepsilon}{a}\right)\sum_{i=1}^{k-j} \left(\lVert v_{\ell_i}a_{i} \rVert + \lVert \pi(a_i, a_{i+1})\rVert+ \lVert v_{\ell_{i+1}}a_{i+1} \rVert\right) \nonumber \\
                & \leq \left(1+\frac{\varepsilon}{a}\right)\sum_{i=1}^{k-j} \left(\lVert \pi(a_i, a_{i+1})\rVert+ \frac{2\varepsilon}{b}(d_i+d_{i+1})\right) \label{eq:mindist} \\
                & \leq \left(1+\frac{\varepsilon}{a}\right)\sum_{i=1}^{k-j} \lVert \pi(a_i, a_{i+1})\rVert+ \left(1+\frac{1}{a}\right)\frac{2\varepsilon}{b}\sum_{i=1}^{k-j}(d_i+d_{i+1}). \label{eq:boundapprox}
	    \end{align}
	    
	    It remains to determine an upper bound for the sum $ \sum_{i=1}^{k-j}(d_i+d_{i+1}) $. Recall that $ d_i $ is the minimum distance from disk $ D_i $ to any other disk $ D_j$. Hence, it follows that
     
	    \begin{align*}
	        d_i+d_{i+1} &\leq \left(\lVert v_{\ell_i}a_i \rVert + \lVert \pi(a_i, a_{i+1}) \rVert\right) + \left(\lVert v_{\ell_{i+1}}a_{i+1}\rVert + \lVert \pi(a_{i+1}, a_{i}) \rVert \right) \\
          & \leq 2\lVert \pi(a_i, a_{i+1}) \rVert + \frac{2}{b}(d_i+d_{i+1}).
	    \end{align*}

        The second inequality in the previous equation comes from the fact that $ \lVert v_{\ell_i}a_i \rVert + \lVert v_{\ell_{i+1}}a_{i+1}\rVert \leq \frac{2\varepsilon}{b}(d_i+d_{i+1}) \leq \frac{2}{b}(d_i+d_{i+1})$, see Equation (\ref{eq:mindist}). Hence, $ d_i+d_{i+1} \leq \frac{2b}{b-2}\lVert \pi(a_i, a_{i+1}) \rVert $. This, when substituted in Equation (\ref{eq:boundapprox}) implies that 
        
        \begin{align*}
            \lVert \pi''(s,t) \rVert & \leq \left(1+\frac{\varepsilon}{a}\right)\sum_{i=1}^{k-j} \lVert \pi(a_i, a_{i+1})\rVert+ \left(1+\frac{1}{a}\right)\frac{2\varepsilon}{b}\cdot\frac{2b}{b-2}\sum_{i=1}^{k-j}\lVert \pi(a_i, a_{i+1})\rVert \\
            & = \left(1+\frac{\varepsilon}{a}+\frac{4\varepsilon}{b-2}+\frac{4\varepsilon}{a(b-2)}\right)\sum_{i=1}^{k-j}\lVert \pi(a_i, a_{i+1})\rVert \\
            & = \left(1+\frac{(b-2)\varepsilon+4a\varepsilon+4\varepsilon}{a(b-2)}\right)\sum_{i=1}^{k-j}\lVert \pi(a_i, a_{i+1})\rVert \\
            & = \left(1+\frac{b\varepsilon+4a\varepsilon+2\varepsilon}{a(b-2)}\right)\sum_{i=1}^{k-j}\lVert \pi(a_i, a_{i+1})\rVert \\
            & = \left(1+\frac{\frac{6a+2}{a-1}\varepsilon+4a\varepsilon+2\varepsilon}{a\left(\frac{6a+2}{a-1}-2\right)}\right)\sum_{i=1}^{k-j}\lVert \pi(a_i, a_{i+1})\rVert \\
            & = (1 + \varepsilon)\cdot \Vert \mathit{SP_w}(s, t) \rVert.
        \end{align*}
        
        Finally, since the length of the shortest path $ \tilde{\pi}(s,t) $ in $ G_\varepsilon $ is at most as large as the length of $ \pi''(s,t) $, we obtain the desired result.
    \end{proof}

    Sometimes, real-world curved objects are represented using polygons; for instance, a disk can be approximated using a regular $ c $-gon, where $ c $ is a sufficiently large value. Then, one approach to solving the WRP on a set of disks would be to approximate each disk with a $ c$-gon, and then use existing algorithms that work on polygons. However, this method is not always optimal.

    \begin{observation}
        \label{prop:cgons}
        Using our discretization scheme for weighted disks provides an approximate shortest path using fewer Steiner points than when using other schemes for triangulations.
    \end{observation}

    \begin{proof}
        First, let us consider that the disks are approximated by regular $ c $-gons circumscribing the disks. We would like to know the value of $ c $ for which the length of a path that coincides with the boundary of the $ c $-gon is a $ (1+\varepsilon) $-approximation of the length of a path that coincides with an arc of a disk. For the rest of the proof we assume that the $ c$-gons are disjoint. Let $ a $ be a corner of the $ c $-gon circumscribing $ D_i $, let $ d $ be the intersection point between the segment~$ \overline{c_ia} $ and $ D_i $. Let $ b $ be the midpoint of an edge containing $ a $, see Figure~\ref{fig:c-gon}.
            \begin{figure}[tb]
                \centering
                \includegraphics{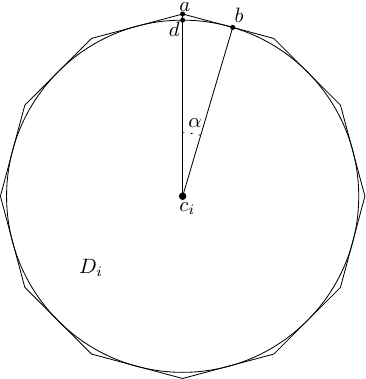}
                \caption{$ c $-gon circumscribing a disk.}
                \label{fig:c-gon}
            \end{figure}
        Now, we want to compare the length of the segment $ \overline{ab} $ with the length of the arc~$ \widehat{db} $. Let~$ \alpha $ be the angle $ \angle{c_ibd} $, and assume, w.l.o.g., that the $ c $-gons have side length~$ 2 $. Then, $ \widehat{db} = \frac{\alpha}{\tan{\alpha}} = \frac{\pi}{c}\cot{\frac{\pi}{c}} $. Hence,
        \begin{equation}
            \label{eq:cgon}
            \frac{|\overline{ab}|}{|\widehat{bd}|} = \frac{1}{\frac{\pi}{c}\cot{\frac{\pi}{c}}} = \frac{\sin{\frac{\pi}{c}}}{\frac{\pi}{c}\cos{\frac{\pi}{c}}} \leq \frac{\frac{\pi}{c}}{\frac{\pi}{c}\left(1-\frac{\left(\frac{\pi}{c}\right)^2}{2}\right)} = \frac{1}{\frac{2-\left(\frac{\pi}{c}\right)^2}{2}} = \frac{2}{2-\left(\frac{\pi}{c}\right)^2},
        \end{equation}
        and if $ c \geq \sqrt{\frac{1+\varepsilon}{2\varepsilon}}\pi $, then $ \frac{|\overline{ab}|}{|\widehat{bd}|} \leq 1+\varepsilon $. The lowest upper bound on the number of vertices of a discretization scheme is obtained in~\cite{aleksandrov2005determining}, giving in this case at most $ \frac{C(P)\sqrt{\frac{1+\varepsilon}{2\varepsilon}}\pi n \log_2{\frac{2}{\varepsilon}}}{\sqrt{\varepsilon}}$ Steiner points, for some parameter $ C(P)>0 $. However, using our approach, we are adding at most $ C(\mathfrak D)\frac{n}{\varepsilon} $, for some other parameter $ C(\mathfrak D) > 0 $. Thus,
        \begin{equation}
            \label{eq:fracpoints}
            \frac{\frac{C(P)\sqrt{\frac{1+\varepsilon}{2\varepsilon}}\pi n \log_2{\frac{2}{\varepsilon}}}{\sqrt{\varepsilon}}}{C(\mathfrak D)\frac{n}{\varepsilon}} \geq \frac{C(P)\pi}{C(\mathfrak D)\sqrt{2}}\sqrt{1+\varepsilon}\log_2{\frac{2}{\varepsilon}}.
        \end{equation}
        We know that $ \sqrt{1+\varepsilon}\log_2{\frac{2}{\varepsilon}} > 1 $, since $ \varepsilon > 0 $, and that $ \frac{C(P)\pi}{C(\mathfrak D)\sqrt{2}} > 0 $, so the value in Equation (\ref{eq:fracpoints}) is at least $ 1 $ for small values of $ \varepsilon $. This concludes the proof that we are adding fewer points than if we approximate each disk with a $ c$-gon, and then we use existing algorithms that work on polygons.
    \end{proof}

One of the reasons for the result in Observation~\ref{prop:cgons} might be because we only have disjoint $ c $-gons on the 2-dimensional space, while most of the discretization schemes we are aware of (see, e.g., ~\cite{aleksandrov1998varepsilon,aleksandrov2000approximation,aleksandrov2005determining,sun2001bushwhack}) are described in terms of a triangulation. In our case, all the triangles inside the $ c $-gons have the same weight, and outside the $ c $-gons we have adjacent triangles all with weight $ 1 $, so we have to add more Steiner points than necessary.
   
\section{A spanner $ G_\theta $}
\label{sec:spanner}

A subgraph $ H $ of a geometric graph $ G $ is a \emph{$t$-spanner} of $ G $, for $ t \geq 1 $, if for each pair of vertices $u$ and $v$, $ \delta_H (u, v) \leq t \cdot \delta_G (u, v)$, where $ \delta_G(u,v) $ is defined as the sum of the weights of the edges along the shortest path between $u$ and~$v$ in $G$. The smallest value $t$ for which $H$ is a $t$-spanner is the \emph{spanning ratio} of $H$. The spanning properties of various geometric graphs have been studied extensively before (see, e.g.,~\cite{bose2013plane,narasimhan2007geometric}). In particular, in this section we study Yao graphs~\cite{flinchbaugh1981strong,yao1982constructing} with the introduction of weighted disks constraints.

A \emph{cone} $ C $ is defined as the region in the plane between two rays originating from a vertex called the \emph{apex} of the cone. When constructing a Yao graph (or $ Y_{2k}$-graph), for each Steiner point $ u $ on the boundary of a disk we consider the rays originating from $u $ with the angle between consecutive rays being $ \theta = \frac{\pi}{k} $. The cones are oriented such that the line containing a ray $ r $ of some cone is tangent to a disk $ D $ at~$ u $. We number the cones in clockwise order around $ u $ starting at~$ r $. The cones around the other vertices have the same orientation as the ones around $ u $. We denote by $ C_i^u $ the $ i $-th cone of a vertex $ u $. For each cone $ C_i $ of each vertex $ u $, we add an edge from $ u $ to the vertex with the shortest Euclidean distance that is visible from $ u $ on the boundary of the disk with the shortest Euclidean distance in that cone. Two vertices are visible if they belong to the boundary of the same disk, or if the line segment joining them does not intersect the interior of any disk. See Figure~\ref{fig:spanner}. If the closest visible vertex from $ u $ is a vertex $ v$ adjacent on the boundary of a disk, we add a circular-arc edge.

\begin{figure}[tb]
    \centering
    \includegraphics{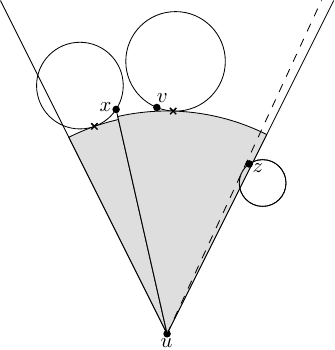}
    \caption{Vertex $ x $ is the closest visible vertex to $ u $ on the closest disk to $ u $.}
    \label{fig:spanner}
\end{figure}

In this section, we show that this Yao graph is a spanner of the graph $ G_\varepsilon $ defined in Section~\ref{sec:discretization}. The proof is similar to the one in~\cite[Theorem $ 8.1 $]{van2014theta}.

\begin{theorem}
    The constrained $ Y_{2k} $-graph ($ k \geq 4 $) is a $ \frac{1}{1-2\sin{\frac{\theta}{2}}} $-spanner of $ G_\varepsilon $.
\end{theorem}

\begin{proof}
    Let $ u $ and $ v $ be two vertices of $ G_\varepsilon $ that can see each other. 
    We show that there exists a weighted path connecting $ u $ and $ v $ in the constrained graph $ Y_{2k} $-graph of length at most $ t\cdot\delta_{G_\varepsilon}(u,v) $, for $ t = \frac{1}{1-2\sin{\frac{\theta}{2}}} $, by induction on the distance between every pair of visible vertices $ u $ and $ v $.

    \paragraph{Base case} Vertices $ u $ and $ v $ are a closest visible pair. If $ (u,v) $ was not an edge of the spanner, then there would be some vertex $ x $ which is joined by an edge of the spanner to $ u $. The vertices $ x $ and $ v $ are connected by an edge in $ G_\varepsilon $ whose length is upper-bounded by the length of a convex chain corresponding to the part of the convex hull of the segment $ \overline{xv} $ that is visible from~$ u $, see Figure~\ref{fig:convexhull}.
    
    \begin{figure}[tb]
    \centering
    \includegraphics{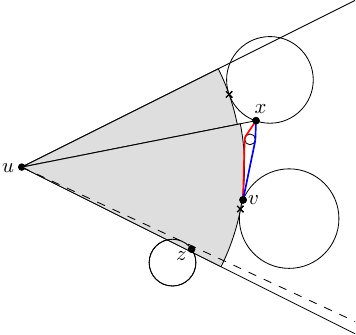}
    \caption{The edge $ (x,v) \in G_\varepsilon $ is represented in blue. The red path from $ x $ to $ v $ corresponds to a convex chain avoiding the disks in $ C^u$.}
    \label{fig:convexhull}
\end{figure}
    
    Since we have at least $ 8 $ cones, the vertex adjacent to $ v $ along this chain is strictly closer to $ v $ than $ u $, contradicting that $ \delta_{G_\varepsilon}(u,v) $ is a closest pair. Hence, since $ v $ is the closest visible vertex, the edge in $ G_\varepsilon $ between $ u $ and $ v $ is an edge in the constrained $ Y_{2k} $-graph and thus there exists a path between $ u $ and $ v $ of length $ \delta_{G_\varepsilon}(u,v) < t \cdot \delta_{G_\varepsilon}(u,v) $.

    \paragraph{Induction step} We assume that the induction hypothesis holds for all pairs of vertices that can see each other and whose distance is less than $ \delta_{G_\varepsilon}(u,v) $. If $ (u,v) $ is an edge in the constrained $ Y_{2k} $-graph, the induction hypothesis follows by the same argument as in the base case. If there is no edge between $ u $ and $ v $, let $ w \in G_\varepsilon $ be the closest visible vertex to $ u $ in the cone of $ u $ that contains $ v $, and let $ w' $ be the closest point to $ u$ in the disk containing $ w $. Let $ x $ be the point along the segment $ \overline{uv} $ such that $ \lvert \overline{uw'} \rvert = \lvert \overline{ux} \rvert $. Since $ x $ lies on $ \overline{uv} $, both $ (u,x) $ and $ (x,v) $ are visibility edges.

    By definition, $ (w',x) $ and $(x,v)$ are visibility edges. Hence, there is a convex chain of visibility edges $ w = p_0, \ldots, p_j = v $ connecting $w$ and $v$. The length of this convex chain can be upper-bounded by the weighted length of the edge $ (w, v) $ in $ G_\varepsilon $. Since we have at least $8$ cones, the length of this edge is strictly less than $|\overline{uv}|$. Hence, since every consecutive pair of vertices along the convex chain can see each other, we can apply induction on each of them. Therefore, there exists a path from $u$ to $v$ via $w$ of length at most
    
    \begin{equation*}
        |\overline{uw}| + t \cdot (|\widehat{ww'}| + |\overline{w'x}|+ |\overline{xv}|).
    \end{equation*}

    Since $ |\overline{uw'}| = |\overline{ux}|$, triangle $\triangle uw'x$ is an isosceles triangle and we can express $|\overline{w'x}|$ as $ 2\sin{\left(\frac{\angle w'ux}{2}\right)} \cdot |\overline{uw'}| \leq 2\sin{\left(\frac{\angle w'ux}{2}\right)} \cdot |\overline{uw}| $. Since this function is increasing for $\angle w'ux \in \left[0, \frac{\pi}{4}\right] $ and $ \angle w'ux $ is at most $ \theta $, it follows that $ |\overline{w'x}| \leq 2\sin{\left(\frac{\theta}{2}\right)}\cdot |\overline{uw}| $. Next we look at $|\overline{xv}|$: since $x$ lies on $\overline{uv}$ and $|\overline{uw'}|= |\overline{ux}|$, it follows that $|\overline{xv}| = |\overline{uv}| - |\overline{ux}| = |\overline{uv}| - |\overline{uw'}|$. Hence, the path between $u$ and~$v$ has length at most
    
    \begin{align*}
        & \ |\overline{uw}| + t \cdot (|\widehat{ww'}| + |\overline{w'x}|+ |\overline{xv}|) \\
        \leq & \ |\overline{uw}| + t \cdot \left(|\widehat{ww'}| + 2\sin{\left(\frac{\theta}{2}\right)} \cdot |\overline{uw}| + |\overline{uv}|-|\overline{uw'}|\right) \\
        = &\ t \cdot |\overline{uv}| + \left(|\overline{uw}|+t|\widehat{ww'}|+2t\sin{\left(\frac{\theta}{2}\right)}\cdot|\overline{uw}| -t|\overline{uw'}|\right).
    \end{align*}

    Hence, for the length of the path to be at most $t\cdot|\overline{uv}|$, we need that
    
    \begin{equation*}
        |\overline{uw}|+t|\widehat{ww'}|+2t\sin{\left(\frac{\theta}{2}\right)}\cdot|\overline{uw}| -t|\overline{uw'}| \leq 0,
    \end{equation*}

    which can be rewritten to
    
    \begin{equation*}
        t \geq \frac{|\overline{uw}|}{|\overline{uw'}|-|\widehat{ww'}|-2\sin{\left(\frac{\theta}{2}\right)}\cdot |\overline{uw}|} \geq \frac{|\overline{uw}|}{|\overline{uw}|-2\sin{\left(\frac{\theta}{2}\right)}\cdot |\overline{uw}|} \geq \frac{1}{1-2\sin{\left(\frac{\theta}{2}\right)}}
    \end{equation*}

    completing the proof.
\end{proof}

The constrained $ Y_{2k} $-graph ($k\geq 4$) is a subgraph of $ G_\varepsilon $ with at most $ C(\mathfrak D)\frac{n}{\varepsilon} $ nodes and at most $ 2kC(\mathfrak D)\frac{n}{\varepsilon} $ edges. Thus, we obtain the following corollary.

\begin{corollary}
    Let $\mathit{SP_w}(s,t)$ be a weighted shortest path between two different points $ s $ and $ t $. There exists a path $ \pi(s,t) $ in the constrained $ Y_{2k} $-graph ($k\geq 4$) such that $ \lVert \pi(s,t) \rVert \leq (1 + \varepsilon)\frac{1}{1-2\sin{\left(\frac{\theta}{2}\right)}}\cdot \Vert \mathit{SP_w}(s, t) \rVert $.
\end{corollary}

\section{Conclusions and open problems}

We presented and analyzed a discretization scheme of the 2D space containing a set of non-overlapping weighted disks. Using this scheme, one can compute an approximate shortest path when the disks on the space have a non-negative weight assigned to them. The main idea of the discretization is to place Steiner points on the boundary of the disks.
        
As future work, it would be interesting to reduce the number of Steiner points that we place on the boundary of the disks, or reduce the number of edges of the associated graph. Finally, a more general version of the problem is to consider some disks that are not mutually disjoint.



\bibliographystyle{abbrv}
\bibliography{bibliography_approximation}

\begin{thebibliography}{10}

\bibitem{aleksandrov1998varepsilon}
L.~Aleksandrov, M.~Lanthier, A.~Maheshwari, and J.-R. Sack.
\newblock An {$\varepsilon$}-approximation algorithm for weighted shortest paths on polyhedral surfaces.
\newblock In {\em Scandinavian Workshop on Algorithm Theory}, pages 11--22. Springer, 1998.

\bibitem{aleksandrov2000approximation}
L.~Aleksandrov, A.~Maheshwari, and J.-R. Sack.
\newblock Approximation algorithms for geometric shortest path problems.
\newblock In {\em Proceedings of the Thirty-Second Annual ACM Symposium on Theory of Computing}, pages 286--295, 2000.

\bibitem{aleksandrov2005determining}
L.~Aleksandrov, A.~Maheshwari, and J.-R. Sack.
\newblock Determining approximate shortest paths on weighted polyhedral surfaces.
\newblock {\em Journal of the ACM}, 52(1):25--53, 2005.

\bibitem{bose2011spanners}
P.~Bose, P.~Carmi, and M.~Couture.
\newblock Spanners of additively weighted point sets.
\newblock {\em Journal of Discrete Algorithms}, 9(3):287--298, 2011.

\bibitem{bose2024steiner}
P.~Bose, G.~Esteban, and A.~Maheshwari.
\newblock A {S}teiner-point-based algorithm for approximate shortest paths in weighted equilateral-triangle meshes.
\newblock {\em Theoretical Computer Science}, page 114583, 2024.

\bibitem{bose2013plane}
P.~Bose and M.~Smid.
\newblock On plane geometric spanners: A survey and open problems.
\newblock {\em Computational Geometry}, 46(7):818--830, 2013.

\bibitem{braid1975synthesis}
I.~C. Braid.
\newblock The synthesis of solids bounded by many faces.
\newblock {\em Communications of the ACM}, 18(4):209--216, 1975.

\bibitem{braid1980stepwise}
I.~C. Braid, R.~C. Hillyard, and I.~A. Stroud.
\newblock Stepwise construction of polyhedra in geometric modelling. {B}rodlie {KW} (ed) mathematical methods in computer graphics and design, 1980.

\bibitem{chang2005shortest}
E.-C. Chang, S.~W. Choi, D.~Kwon, H.~Park, and C.~K. Yap.
\newblock Shortest path amidst disc obstacles is computable.
\newblock In {\em Proceedings of the Twenty-First Annual Symposium on Computational Geometry}, pages 116--125, 2005.

\bibitem{chen2013computing}
D.~Z. Chen, J.~Hershberger, and H.~Wang.
\newblock Computing shortest paths amid convex pseudodisks.
\newblock {\em SIAM Journal on Computing}, 42(3):1158--1184, 2013.

\bibitem{chen2015computing}
D.~Z. Chen and H.~Wang.
\newblock Computing shortest paths among curved obstacles in the plane.
\newblock {\em ACM Transactions on Algorithms}, 11(4):1--46, 2015.

\bibitem{chew1985planning}
L.~P. Chew.
\newblock Planning the shortest path for a disc in $ {O}(n^2\log{n}) $ time.
\newblock In {\em Proceedings of the First Annual Symposium on Computational Geometry}, pages 214--220, 1985.

\bibitem{connolly1987application}
M.~L. Connolly.
\newblock An application of algebraic topology to solid modeling in molecular biology.
\newblock {\em The Visual Computer}, 3:72--81, 1987.

\bibitem{de2014note}
J.-L. De~Carufel, C.~Grimm, A.~Maheshwari, M.~Owen, and M.~Smid.
\newblock A note on the unsolvability of the weighted region shortest path problem.
\newblock {\em Computational Geometry}, 47(7):724--727, 2014.

\bibitem{floriani}
L.~de~Floriani, P.~Magillo, and E.~Puppo.
\newblock Applications of computational geometry to geographic information systems.
\newblock {\em Handbook of Computational Geometry}, 7:333--388, 2000.

\bibitem{dobkin1990computational}
D.~P. Dobkin and D.~L. Souvaine.
\newblock Computational geometry in a curved world.
\newblock {\em Algorithmica}, 5(1-4):421--457, 1990.

\bibitem{flinchbaugh1981strong}
B.~Flinchbaugh and L.~Jones.
\newblock Strong connectivity in directional nearest-neighbor graphs.
\newblock {\em SIAM Journal on Algebraic Discrete Methods}, 2(4):461--463, 1981.

\bibitem{forrest1986invited}
A.~Forrest.
\newblock Invited talk on computational geometry and software engineering.
\newblock In {\em ACM Symposium on Computational Geometry}, 1986.

\bibitem{gaw}
D.~Gaw and A.~Meystel.
\newblock Minimum-time navigation of an unmanned mobile robot in a 2-1/2{D} world with obstacles.
\newblock In {\em Proceedings of the 1986 IEEE International Conference on Robotics and Automation}, volume~3, pages 1670--1677. IEEE, 1986.

\bibitem{gewali1988path}
L.~Gewali, A.~Meng, J.~S.~B. Mitchell, and S.~Ntafos.
\newblock Path planning in 0/1/weighted regions with applications.
\newblock In {\em Proceedings of the fourth annual symposium on Computational geometry}, pages 266--278, 1988.

\bibitem{hershberger1999optimal}
J.~Hershberger and S.~Suri.
\newblock An optimal algorithm for {E}uclidean shortest paths in the plane.
\newblock {\em SIAM Journal on Computing}, 28(6):2215--2256, 1999.

\bibitem{hershberger2022near}
J.~Hershberger, S.~Suri, and H.~Yildiz.
\newblock A near-optimal algorithm for shortest paths among curved obstacles in the plane.
\newblock {\em SIAM Journal on Computing}, 51(4):1296--1340, 2022.

\bibitem{kamphuis}
A.~Kamphuis, M.~Rook, and M.~H. Overmars.
\newblock Tactical path finding in urban environments.
\newblock In {\em First International Workshop on Crowd Simulation}. Citeseer, 2005.

\bibitem{kim2004shortest}
D.-S. Kim, K.~Yu, Y.~Cho, D.~Kim, and C.~Yap.
\newblock Shortest paths for disc obstacles.
\newblock {\em Lecture Notes in Computer Science}, pages 62--70, 2004.

\bibitem{KirkpatrickL16-cccg}
D.~G. Kirkpatrick and P.~Liu.
\newblock Characterizing minimum-length coordinated motions for two discs.
\newblock In {\em Proceedings of the 28th Canadian Conference on Computational Geometry}, pages 252--259, 2016.

\bibitem{lienhardt1991topological}
P.~Lienhardt.
\newblock Topological models for boundary representation: a comparison with n-dimensional generalized maps.
\newblock {\em Computer-Aided Design}, 23(1):59--82, 1991.

\bibitem{maheshwari2007n}
A.~Maheshwari, D.~Nussbaum, J.-R. Sack, and J.~Yi.
\newblock An ${O}(n^2\log{n})$ time algorithm for computing shortest paths amidst growing discs in the plane.
\newblock In {\em Algorithms and Computation: 18th International Symposium}, pages 668--680. Springer, 2007.

\bibitem{Mitchell1}
J.~S.~B. Mitchell.
\newblock Shortest paths and networks.
\newblock In J.~E. Goodman, J.~O'Rourke, and C.~D. Toth, editors, {\em Handbook of Discrete and Computational Geometry, Second Edition}, pages 811--848. Chapman and Hall/CRC, 2017.

\bibitem{mitchell1991weighted}
J.~S.~B. Mitchell and C.~H. Papadimitriou.
\newblock The weighted region problem: finding shortest paths through a weighted planar subdivision.
\newblock {\em Journal of the ACM}, 38(1):18--73, 1991.

\bibitem{narasimhan2007geometric}
G.~Narasimhan and M.~Smid.
\newblock {\em Geometric spanner networks}.
\newblock Cambridge University Press, 2007.

\bibitem{pavlidis1983curve}
T.~Pavlidis.
\newblock Curve fitting with conic splines.
\newblock {\em ACM Transactions on Graphics}, 2(1):1--31, 1983.

\bibitem{pocchiola1995computing}
M.~Pocchiola and G.~Vegter.
\newblock Computing the visibility graph via pseudo-triangulations.
\newblock In {\em Proceedings of the eleventh annual symposium on Computational geometry}, pages 248--257, 1995.

\bibitem{pratt1985techniques}
V.~Pratt.
\newblock Techniques for conic splines.
\newblock {\em ACM SIGGRAPH Computer Graphics}, 19(3):151--160, 1985.

\bibitem{requicha1980representations}
A.~G. Requicha.
\newblock Representations for rigid solids: Theory, methods, and systems.
\newblock {\em ACM Computing Surveys}, 12(4):437--464, 1980.

\bibitem{rowe}
N.~C. Rowe and R.~S. Ross.
\newblock Optimal grid-free path planning across arbitrarily contoured terrain with anisotropic friction and gravity effects.
\newblock {\em IEEE Transactions on Robotics and Automation}, 6(5):540--553, 1990.

\bibitem{Sharir}
M.~Sharir and S.~Sifrony.
\newblock Coordinated motion planning for two independent robots.
\newblock {\em Annals of Mathematics and Artificial Intelligence}, 3(1):107--130, 1991.

\bibitem{smid2021improved}
M.~Smid.
\newblock An improved construction for spanners of disks.
\newblock {\em Computational Geometry}, 92:101682, 2021.

\bibitem{smith1986invited}
A.~Smith.
\newblock Invited talk on the complexity of images in the movies.
\newblock In {\em ACM Symposium on Computational Geometry}, 1986.

\bibitem{sturtevant2}
N.~R. Sturtevant, D.~Sigurdson, B.~Taylor, and T.~Gibson.
\newblock Pathfinding and abstraction with dynamic terrain costs.
\newblock In {\em Proceedings of the AAAI Conference on Artificial Intelligence and Interactive Digital Entertainment}, volume~15, pages 80--86, 2019.

\bibitem{sun2001bushwhack}
Z.~Sun and J.~H. Reif.
\newblock On finding approximate optimal paths in weighted regions.
\newblock {\em Journal of Algorithms}, 58(1):1--32, 2006.

\bibitem{van2008planning}
J.~Van Den~Berg and M.~Overmars.
\newblock Planning the shortest safe path amidst unpredictably moving obstacles.
\newblock In {\em Algorithmic Foundation of Robotics VII: Selected Contributions of the Seventh International Workshop on the Algorithmic Foundations of Robotics}, pages 103--118. Springer, 2008.

\bibitem{van2014theta}
A.~Van~Renssen.
\newblock {\em Theta-Graphs and Other Constrained Spanners}.
\newblock PhD thesis, Carleton University, 2014.

\bibitem{yao1982constructing}
A.~C.-C. Yao.
\newblock On constructing minimum spanning trees in k-dimensional spaces and related problems.
\newblock {\em SIAM Journal on Computing}, 11(4):721--736, 1982.

\end{thebibliography}

\end{document}